\documentclass[11pt]{article}
\usepackage{amssymb,amsthm,amsmath,graphicx,color}
\usepackage{amsfonts,natbib}
\usepackage[caption=false]{subfig}
\usepackage{epsfig,latexsym,graphicx}
\newtheorem{theorem}{Theorem}[section]

\newtheorem{lemma}[theorem]{Lemma}

\begin{document}

\title{\bf{Improvements in the Small Sample Efficiency of the Minimum $S$-Divergence Estimators under Discrete Models}\footnote{This is a part of the Ph.D. dissertation of the first author.}}
\author{Abhik Ghosh and Ayanendranath Basu\thanks{Corresponding author. Email: ayanbasu@isical.ac.in\vspace{2pt}}
	\\
	\vspace{6pt}  
{\em{Indian Statistical Institute, Kolkata, India}}
}
\maketitle

\begin{abstract}
This paper considers the problem of inliers and empty cells 
and the resulting issue of relative inefficiency in estimation 
under pure samples from a discrete population when the sample size is small. 
Many minimum divergence estimators in the $S$-divergence family,
although possessing very strong outlier stability properties,
often have very poor small sample efficiency in the presence of inliers
and some are not even defined in the presence of a single empty cell;
this limits the practical applicability of these estimators,
in spite of their otherwise sound robustness properties and high asymptotic efficiency.
Here, we will study a penalized version of the $S$-divergences 
such that the resulting minimum divergence estimators are free from these issues   
without altering their robustness properties and asymptotic efficiencies. 
We will give a general proof for the asymptotic properties of these minimum penalized $S$-divergence estimators.
This provides a significant addition to the literature as the asymptotics
of penalized divergences which are not finitely defined are currently unavailable in the literature. 
The small sample advantages of the minimum penalized $S$-divergence estimators 
are examined through an extensive simulation study 
and some empirical suggestions regarding the choice of the relevant underlying tuning parameters
are also provided.
\end{abstract}

\section{Introduction}

Minimum divergence inference provides an excellent theoretical alternative to the
classical maximum likelihood approach in presence of contamination in the observed data.
Many minimum divergence estimators are highly robust in the presence of outliers 
and have asymptotic efficiencies close to that of the maximum likelihood estimator under the 
pure model. Indeed, some minimum divergence estimators, along with their high robustness, 
provide full asymptotic efficiency under the true model (e.g., those based on the class of disparities). 
Although this is a very desirable large sample asymptotic property,
the results are not always so spectacular when applied to practical real-life data-sets of small sizes.  
Some of the robust minimum divergence estimators have very poor performance compared to 
the maximum likelihood estimators in small samples under pure data. Examples of such divergences include 
the celebrated Hellinger distance along with other Cressie-Read power divergences \citep{Cressie/Read:1984}
with large negative values of the tuning parameter.
The mean square error of these estimators at small sample sizes often turn out to be  substantially higher than that 
of the maximum likelihood estimator under pure model. 
This limits the use of such estimators in spite of their demonstrated strong robustness properties and good asymptotic performances.

 The issue of small sample efficiency of  the robust minimum divergence estimators
 has received some attention in the recent literature. The root of this problem appears to be 
 the presence of the so-called ``inliers" in the data. 
 Inliers are those values in the sample space where fewer observations are available compared to what is expected under the model. 
An empty cell is the most extreme case of an inlier. 
The inlier problem becomes more acute as the sample size becomes smaller. 
 Since most of the robust density based minimum divergence estimators successfully deal 
 with the outliers by down-weighting those observations by the model density, 
 they in turn magnify the effect of inliers. 
 Hence weights attached to the inliers or empty cells play a crucial role in the poor performance of the estimators
 in small samples. \cite{Lindsay:1994} observed this phenomenon in case of the popular Hellinger distance. 
 The problem of inliers can be further understood by noting the fact that 
 minimum divergence estimators with suitable treatments of inliers provide competitive 
 small sample performance compared to the maximum likelihood estimator under the true model. 
 Examples of such divergences include, among others, the Cressie-Read power divergence with positive values of the tuning parameter, 
 the negative exponential disparity \citep{Lindsay:1994, Basu/Sarkar/Vidyashankar:1997} and 
the generalizations \citep{Bhandari/Basu/Sarkar:2006} of the negative exponential disparity.

 Although the concept of inlier is relatively new compared to that of the outlier, 
 there has been a fair bit of recent activity leading to several methods for inlier correction, 
 without compromising the robustness properties of the corresponding minimum divergence estimators. 
 \cite{Basu/etc:2011} and \cite{Mandal/Basu:2010a} provide a comprehensive description of the concept and 
 relevant approaches to solve the problem of inliers. 
 Among all the available methods, in this paper, we will consider one particular technique based on
 the method of penalized divergences and use it to improve the minimum $S$-divergence estimators. 
 The $S$-divergence family has been developed in \citet{Ghosh/etc:2013a,Ghosh/etc:2013} and 
 generates many robust estimators without any significant loss in efficiency. 
 This large family includes the popular Cressie-Read power divergence and 
 the density power divergence \citep{Basu/etc:1998} measures as special cases. \cite{Ghosh:2013} and \cite{Ghosh/Basu:2015}
 have also derived the asymptotic distribution of the minimum $S$-divergence estimators 
 under discrete and continuous models, respectively. 
 However, just like many other density based divergences, 
 the $S$-divergences also use the model density to down-weight the outliers and 
 their small sample performance becomes worse in presence of inliers under the pure model 
 as will be demonstrated later in the paper. 
We will provide a modification to the minimum $S$-divergence estimators 
 using the concept of penalized divergences and prove their asymptotic equivalence to 
 the original minimum $S$-divergence estimator. 
 The corresponding estimator will also be robust under data contamination with improved efficiency in small samples. 
 
It is important to clearly spell out what is new in the present paper. 
\cite{Mandal/Basu:2010} established the asymptotic equivalence of the minimum divergence estimators 
corresponding to ordinary and penalized disparities. 
However, their proof was restricted to the cases where the ordinary divergences are finitely defined with probability one.
This excludes all divergences within the Cressie-Read family of disparities 
for which the tuning parameter $\lambda\leq -1$ (as well as many other disparities outside the Cressie-Read family).
Thus, although \cite{Mandal/Basu:2010} could define penalized versions of disparities 
like, say,  the Kullback-Leibler divergence ($\lambda=-1$) or the Neyman's chi-square ($\lambda=-2$), 
they did not have a proof of the asymptotic normality of the corresponding minimum divergence estimator.
The approach of our proof transcends this limitation.
We will present a general proof applicable to all the divergences within the $S$-divergence family.
Our proof can be easily generalized to all disparities, and also accommodates the class of density power divergences.
Thus, not only we allow the controlling of all disparities, including those which are not finitely defined,
we also add another dimension to this exercise by including the divergences within the $S$-divergence family, 
and in particular the members of the density power divergence family.
We will, however, restrict our attention to discrete models throughout the paper, 
as this is the case where the empty cells are more relevant.

Another major contribution of the present paper is to study the small sample behaviors 
of different minimum $S$-divergence estimators and their penalized versions, to be introduced here,
through extensive simulations under the Poisson model.     
The study of the MSDEs in small samples indicates the necessity of inlier correction 
for many robust members within the $S$-divergence family. 
As a solution, we then consider a penalized version of the $S$-divergence measure
and empirically illustrate their small sample superiority in inlier control.
Indeed, for this purpose, we define the penalized $S$-divergence by replacing 
the weights attached to the empty cells in the $S$-divergence by a suitably chosen penalty factor.
The choice of this penalty factor becomes crucial for the improvement of their  small sample efficiency. 
A large scale simulation exercise studies this problem in great detail and attempts 
to find out the optimum value of the penalty factor separately for each member of the divergence family
over different (small) sample sizes and different model parameters.
Some overall practical suggestions and guidelines are also provided for practitioners through proper empirical evidences.
Another possible intuitive extension of the penalty scheme is also proposed at the end of the paper with some brief suggestions.  
A real data illustration is also provided.
 
The rest of the paper is organized as follows. We begin with a brief description of the 
minimum $S$-divergence estimators in Section \ref{SEC:MSDE} and 
show how the members of the $S$-divergence family are 
affected by the inliers in terms of their small sample efficiency compared to the maximum
likelihood estimator. Then we introduce the concept of ``penalized $S$-divergence" and the 
corresponding minimum divergence estimators in Section \ref{SEC:5MPSDE} and 
prove their asymptotic equivalence to the original minimum $S$-divergence estimator in Section \ref{SEC:5MPSDE_asymp_discrete}. 
We illustrate the performance of the minimum penalized $S$-divergence estimators in Section \ref{SEC:choice_h}
through an extensive simulation study, where we suggest suitable optimum choices 
of the penalty factor for practical application of different penalized estimators at small sample sizes.   
A real data example is considered in Section \ref{SEC:examples}.
Finally we end the paper with some conclusions, recommendations and discussions on possible future extensions 
in Section \ref{SEC:discussion}.


\section{The Minimum $S$-Divergence Estimators (MSDE) under Discrete Models and its Small Sample Efficiency}
\label{SEC:MSDE}

	The $S$-Divergence family has been defined as a general family of divergence measures
	including the famous Cressie-Read power divergence family and the density power divergence family as 
	its subclasses \citep{Ghosh/etc:2013a,Ghosh/etc:2013}. It is defined in terms of two parameters 
	$\alpha\geq 0$ and $\lambda\in \mathbb{R}$ as
\begin{eqnarray}
		S_{(\alpha, \lambda)}(g,f) &=&  \frac{1}{A} ~ \int ~ f^{1+\alpha}  -   \frac{1+\alpha}{A B} ~ \int ~~ f^{B} g^{A}  + \frac{1}{B} ~ \int ~~ g^{1+\alpha}, 
 \label{EQ:S_div}
\end{eqnarray}
where $A = 1+\lambda (1-\alpha)$ and $B = \alpha - \lambda (1-\alpha)$.
For $A=0$ the $S$-divergence measures may be re-defined by its continuous limit as $A \rightarrow 0$ 
so that
\begin{eqnarray}
    S_{(\alpha,\lambda : A = 0)}(g,f) &=& \lim_{A \rightarrow 0} ~ S_{(\alpha, \lambda)}(g,f) 
    =  \int f^{1+\alpha} \log\left(\frac{f}{g}\right) - \int \frac{(f^{1+\alpha} - g^{1+\alpha})}{{1+\alpha}}.
\label{EQ:S_div_A0}
\end{eqnarray}
Similarly, for $B=0$, we have
\begin{eqnarray}
    S_{(\alpha,\lambda : B = 0)}(g,f) = \lim_{B \rightarrow 0} ~ S_{(\alpha, \lambda)}(g,f) 
    =  \int g^{1+\alpha} \log\left(\frac{g}{f}\right) - \int \frac{(g^{1+\alpha} - f^{1+\alpha})}{{1+\alpha}}.
    \label{EQ:S_div_B0}
\end{eqnarray}
	Note that at $\alpha = 0$ , the $S$-divergence family reduces to the Cressie-Read family 
	having parameter  $\lambda$ and at $\alpha =1$, it  becomes independent of $\lambda$ coinciding 
	with the $L_2$ divergence. On the other hand, at $\lambda = 0$, it generates the density dower divergences with  parameter $\alpha$.  The members of the $S$-divergence family 
	are indeed genuine statistical divergence measures provided $\lambda \in \mathbb{R}$, and 
	$\alpha \geq 0$.


We will consider the set-up for parametric estimation with discrete model families. 
	We have $n$ independent and identically distributed observations  $X_1$, $\cdots$, $X_n$ from the true
	 distribution $G$ having probability mass function (pmf) $g$. 
	 Without loss of generality, the support of $g$ is assumed to be $\chi = \{0, 1, 2, \cdots \}$. 
	 We want to model it by a parametric family of model pmf
	 $\mathcal{F}=\{f_\theta : \theta \in \Theta \subseteq \mathbb{R}^p\}$.
	 Then the $S$-divergence measure between the data and the model is defined through the relative
	  frequency vector $\mathbf{r}_n = ( r_n(0), r_n(1), \cdots )^T$ and the model probability vector
	  $\mathbf{f}_\theta = ( f_\theta(0), f_\theta(1), \cdots )^T$; here for any $x\in \chi$, we define
	  $r_n(x) = \frac{1}{n} \sum_{i=1}^n I(X_i = x)$ with $I(E)$ being the indicator function of the event $E$. 
	  The minimum $S$-divergence estimator  is the parameter value which minimizes 
	  the $S$-divergence measure between the data $\mathbf{r}_n$ and the model $\mathbf{f}_\theta$. 
	  Hence, the estimating equation for the minimum $S$-divergence estimator is given by
\begin{eqnarray}
    \sum_{x=0}^\infty f_{\theta}^{1+\alpha}(x) u_{\theta}(x) - \sum_{x=0}^\infty f_{\theta}^{B}(x) r_n^{A}(x) u_{\theta}(x) &=& 0, \\
\mbox{or,   } ~~ \sum_{x=0}^\infty K(\delta(x))f_{\theta}^{1+\alpha}(x) u_{\theta}(x) = 0,
\label{EQ:S-est_eqn_discrete}
\end{eqnarray}
where $\delta(x)= \delta_n(x) = \frac{r_n(x)}{f_{\theta}(x)} - 1$, $ K(\delta) = \frac{(\delta+1)^A - 1}{A}$ and $u_\theta(x)=\nabla \ln f_\theta(x)$ is the likelihood score function. 
Here,  $\nabla = (\nabla_1, \ldots, \nabla_p)^T$ denotes the derivative with respect to 
$\theta=(\theta_1, \ldots, \theta_p)^T$. 
See \cite{Ghosh:2013}  and \cite{Ghosh/etc:2013} for detailed properties of the 
minimum $S$-divergence estimator under discrete models, including their asymptotic distribution and 
the influence function.

\begin{figure}[!th]
\centering
\subfloat[$\theta=3$, $\lambda=0$]{
\includegraphics[width=0.3\textwidth]{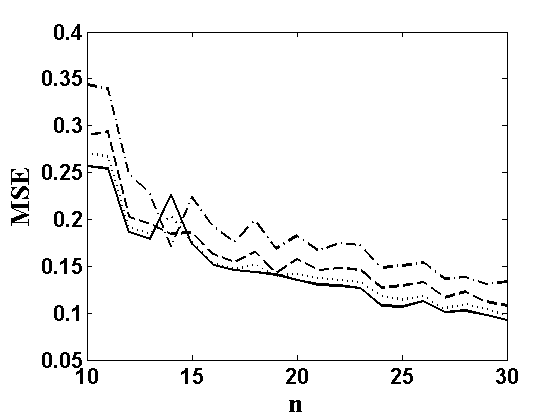}
\label{fig:MSDE_Poiss3_L0}}
~ 
\subfloat[$\theta=3$, $\lambda=-0.5$]{
\includegraphics[width=0.3\textwidth]{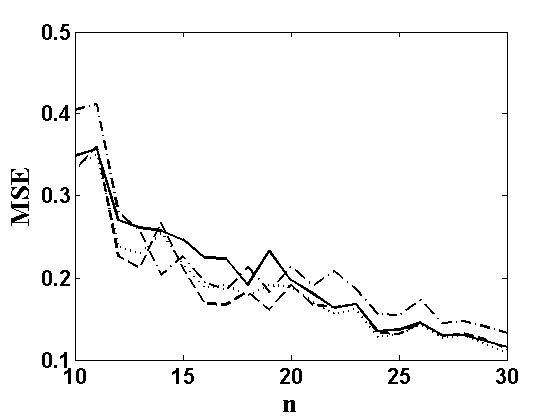}
\label{fig:MSDE_Poiss3_L-5}}
~ 
\subfloat[$\theta=3$, $\lambda=-1$]{
\includegraphics[width=0.3\textwidth]{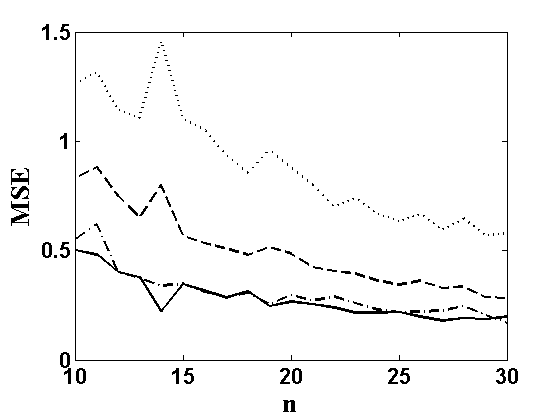}
\label{fig:MSDE_Poiss3_L-10}}
\\ 
\subfloat[$\theta=5$, $\lambda=0$]{
\includegraphics[width=0.3\textwidth]{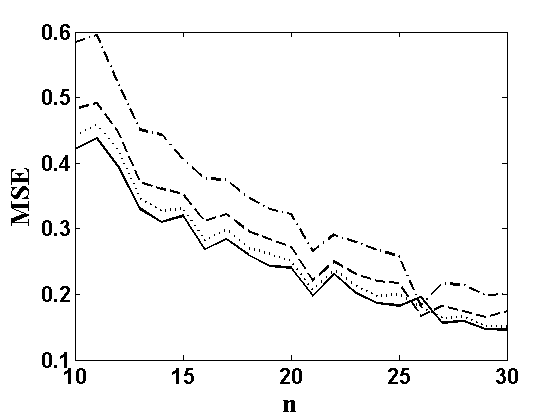}
\label{fig:MSDE_Poiss5_L0}}
~ 
\subfloat[$\theta=5$, $\lambda=-0.5$]{
\includegraphics[width=0.3\textwidth]{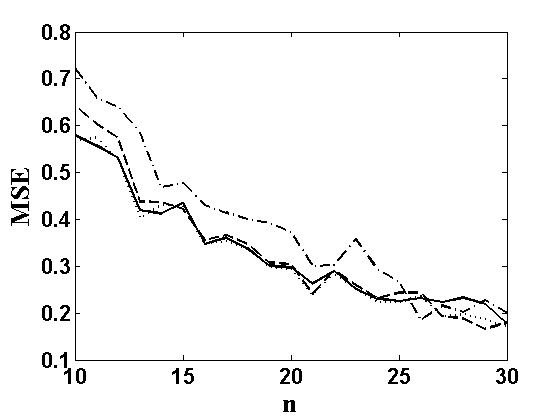}
\label{fig:MSDE_Poiss5_L-5}}
~ 
\subfloat[$\theta=5$, $\lambda=-1$]{
\includegraphics[width=0.3\textwidth]{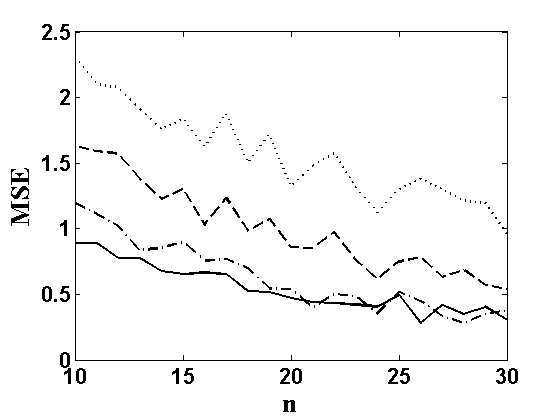}
\label{fig:MSDE_Poiss5_L-10}}
\\ 
\subfloat[$\theta=8$, $\lambda=0$]{
\includegraphics[width=0.3\textwidth]{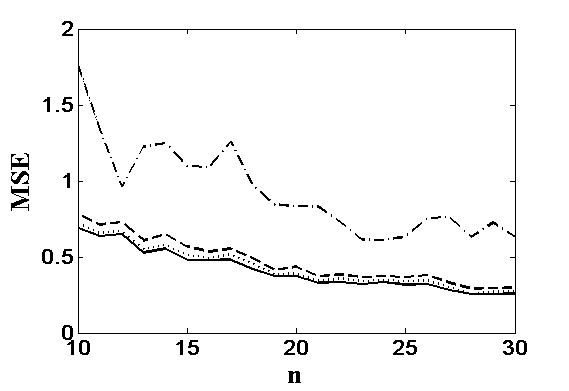}
\label{fig:MSDE_Poiss8_L0}}
~ 
\subfloat[$\theta=8$, $\lambda=-0.5$]{
\includegraphics[width=0.3\textwidth]{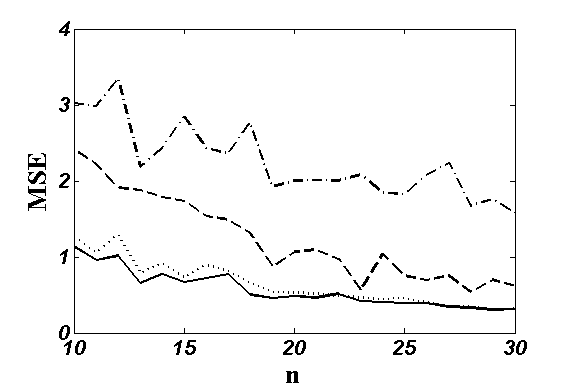}
\label{fig:MSDE_Poiss8_L-5}}
~ 
\subfloat[$\theta=8$, $\lambda=-1$]{
\includegraphics[width=0.3\textwidth]{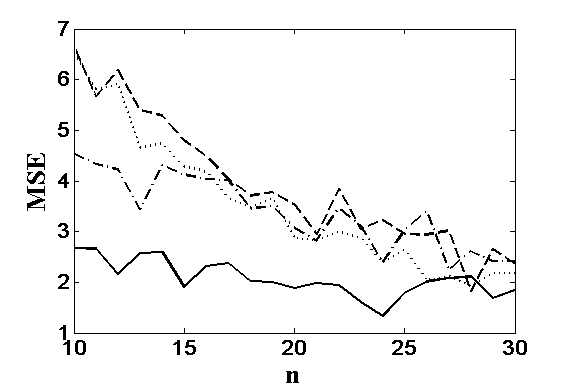}
\label{fig:MSDE_Poiss8_L-10}}
\caption[MSE of $\sqrt{n}\widehat{\theta}_{\alpha, \lambda}$ over sample size $n$ for different $\lambda$, $\alpha$ and $\theta$]{MSE of $\sqrt{n}\widehat{\theta}_{\alpha, \lambda}$ over sample size $n$ for different $\lambda$, $\alpha$ and $\theta$ [Dotted line: $\alpha=0.1$; Dashed line: $\alpha=0.25$; Dot-dashed line: $\alpha=0.5$; Solid line: 
MLE ($\lambda = 0$, $\alpha = 0$) in the first column; minimum Hellinger distance estimator ($\lambda = -0.5$, $\alpha = 0$) in the second column; minimum $L_2$ distance estimator ($\alpha = 1$) in the third column. ] }
\label{FIG:5MSE_MSDE}
\end{figure}

Let us now examine the small sample performance of the minimum $S$-divergence estimators (MSDE), say 
$\widehat{\theta}_{\alpha, \lambda}$, for different values of tuning parameters $\alpha$ and $\lambda$. 
We consider the discrete Poisson model with mean $\theta$ and perform a simulation study 
to examine empirical MSE of the MSDEs under several small sample sizes.
Figure \ref{FIG:5MSE_MSDE} below shows the MSE of $\sqrt{n}\widehat{\theta}_{\alpha, \lambda}$ 
(=$n\times $MSE of $\widehat{\theta}_{\alpha, \lambda}$) over sample size $n$ 
for different values of the Poisson parameter $\theta$. 
Note that, according to the asymptotic theory of MSDEs, this MSE of the $\sqrt{n}$ times MSDEs 
converges to a constant limit depending only on the tuning parameter $\alpha$ under the (pure) model
and increases as $\alpha$ increases.  

However, interestingly, it is observed from Figure \ref{FIG:5MSE_MSDE} that the MSE of 
$\sqrt{n}\widehat{\theta}_{\alpha, \lambda}$ increases significantly as the sample size $n$ decreases
for $\lambda=-0.5, -1$ and smaller $\alpha$; 
it implies that the MSDEs corresponding to those tuning parameters 
(negative $\lambda$ and smaller $\alpha$) become unstable at the small sample size.
The main reason behind this instability in fact comes from the existence of inliers and 
empty cells in the sample with smaller sizes; this fact is also justified by the fact that 
the instability of those MSDEs increases for large values of $\theta$ 
where the chances of an inlier or empty cell is higher.
However, it is already observed in \cite{Ghosh/etc:2013} that these MSDEs
with negative $\lambda$ are highly robust in presence of outliers and 
they cannot be ignored due to their strong robustness and good asymptotic properties; 
but their application to small sample becomes restricted due to its inlier problem.
Thus, we need to have a small sample correction on the MSDEs to control the inliers 
keeping all the good robustness and asymptotic properties as it is.

Further, all the $S$-divergence measures with $\lambda<-1$, any $0\leq \alpha<1$ and $\lambda=-1, \alpha=0$ 
are in fact becomes undefined if there is an empty cell in the sample from a discrete population. 
This can be seen by looking at the explicit form of the corresponding $S$-divergence given by 
\begin{eqnarray}
S_{(\alpha, \lambda)}(\mathbf{r}_n,\mathbf{f}_\theta)  
=  \frac{1}{A} ~ \sum_{x=0}^\infty f_{\theta}^{1+\alpha}(x) 
- \frac{1+\alpha}{A B} ~\sum_{x=0}^\infty f_{\theta}^{B}(x) r_n^{A}(x) 
+ \frac{1}{B} ~\sum_{x=0}^\infty r_n^{1+\alpha}(x),
\label{EQ:5S-div_discrete}
\end{eqnarray}
where $A$ and $B$ are non-zero. 
Now, whenever $A<0$ (i.e., $\lambda< -1/(1-\alpha)$, see Table \ref{TAB:A}), then the term containing $r_n(x)^A$ 
(that also contains $f_\theta$ and so cannot be neglected in the objective function or the estimating equation) 
becomes undefined at those points $x$ in the sample space for which $r_n(x)=0$ (empty cells) 
and hence the corresponding divergence also becomes undefined in the presence of even a single empty cell. 
Thus using such divergences for the derivation of the minimum divergence estimator becomes a fruitless exercise. 
The same can be observed for the case $A=0$ also,
since then the $S$-divergence measure contains a term involving $f_\theta$ and $\log(r_n)$ 
which becomes undefined at $r_n(x)=0$. All these motivate for
a suitable inlier correction in the minimum $S$-divergence estimator.

\section{The Penalized $S$-Divergence (PSD) and Minimum Divergence Estimation}
\label{SEC:5MPSDE}

The concept of penalized divergence was used in the context of successful inlier correction 
by \cite{Mandal/Basu:2010} and \cite{Mandal/Basu:2010a} to modify the Cressie-Read power divergence family, 
a particular subfamily of the $S$-divergences. 
We will now extend it in the case of general family of $S$-divergence and 
examine its performance with respect to efficiency, robustness and inlier controls.

Consider the set-up of discrete parametric model as described in previous section. 
Then the $S$-divergence measure between the data and the model for $A\ne 0$ and $B\ne 0$ 
is given by (\ref{EQ:5S-div_discrete}), which can be further re-written as 
\begin{eqnarray}
S_{(\alpha, \lambda)}(\mathbf{r}_n,\mathbf{f}_\theta)   &=& ~ \sum_{x:r_n(x)\ne 0}~
\left[\frac{1}{A}f_\theta^{1+\alpha}(x)-\frac{1+\alpha}{AB}~f_\theta^{B}(x)r_n^{A}(x)
+\frac{1}{B}~r_n^{1+\alpha}(x)\right] 
\nonumber\\ && + \sum_{x:r_n(x)=0}~ 
\left[\frac{1}{A}f_\theta^{1+\alpha}(x)-\frac{1+\alpha}{AB}~f_\theta^{B}(x)r_n^{A}(x)
+\frac{1}{B}~r_n^{1+\alpha}(x)\right].
\end{eqnarray}
Note that the first term is always defined, 
whereas the second term is not defined at negative values of $A$.
For the positive values of $A$, the second term further simplifies to 
$\frac{1}{A} ~ \sum_{x:r_n(x)=0} f_\theta^{1+\alpha}(x)$ which can also be defined with $A<0$.
Thus, motivating from this fact and retaining similar consistency in the expression of 
the divergence, we define the penalized version of the $S$-divergence under discrete models as
\begin{eqnarray}
PSD_{(\alpha, \lambda)}^h(\mathbf{r}_n,\mathbf{f}_\theta) &=& ~ \sum_{x:r_n(x)\ne 0} ~ 
\left[ \frac{1}{A} f_\theta^{1+\alpha}(x)  -   \frac{1+\alpha}{A B} ~ f_\theta^{B}(x) r_n^{A}(x)  
+ \frac{1}{B} ~ r_n^{1+\alpha}(x) \right] \nonumber\\
&& ~~  +  h \sum_{x:r_n(x)=0} f_\theta^{1+\alpha}(x), 
\label{EQ:PSD}
\end{eqnarray}
where $h$ can be thought of as a penalty factor. 
In the absence of any empty cell, the penalized $S$-divergence (PSD) coincides with the ordinary $S$-divergence;
also in the presence of empty cells it coincides with the ordinary $S$-divergence 
only for the particular choice $h=\frac{1}{A}$ provided $A>0$ (see Table \ref{TAB:A}).
Thus the PSD modification makes  the $S$-divergence finitely defined for all $A\in \mathbb{R}$
and adjusts  the weights of empty cells by the factor $h$ for the cases $A>0$.
Note that the function $PSD_{(\alpha, \lambda)}^h(\mathbf{r}_n,\mathbf{f}_\theta)$  as defined above 
is a genuine statistical divergence for all $(\alpha,\lambda)$ with $A\ne 0$ and $B\ne 0$ and $h \geq 0$.

\begin{table}[!th]
\centering 
\caption{Values of $A$ (and empty cell weight $\frac{1}{A}$ in parenthesis) for different $\alpha$ and $\lambda$ }
\resizebox{0.9\textwidth}{!}{
\begin{tabular}{r|rrrrrrr} \hline
$\lambda$ & \multicolumn{7}{|c}{$\alpha$}\\
	&	0				&	0.1				&	0.25				&	0.4				&	0.5				&	0.7				&	1				\\\hline
$-$2	&	$-$1	(	$-$1.00	)	&	$-$0.80	(	$-$1.25	)	&	$-$0.50	(	$-$2.00	)	&	$-$0.20	(	$-$5.00	)	&	0.00	(	Inf	)	&	0.40	(	2.50	)	&	1	(	1	)	\\
$-$1.5	&	$-$0.5	(	$-$2.00	)	&	$-$0.35	(	$-$2.86	)	&	$-$0.13	(	$-$8.00	)	&	0.10	(	10.00	)	&	0.25	(	4.00	)	&	0.55	(	1.82	)	&	1	(	1	)	\\
$-$1	&	0	(	Inf	)	&	0.10	(	10.00	)	&	0.25	(	4.00	)	&	0.40	(	2.50	)	&	0.50	(	2.00	)	&	0.70	(	1.43	)	&	1	(	1	)	\\
$-$0.5	&	0.5	(	2.00	)	&	0.55	(	1.82	)	&	0.63	(	1.60	)	&	0.70	(	1.43	)	&	0.75	(	1.33	)	&	0.85	(	1.18	)	&	1	(	1	)	\\
$-$0.1	&	0.9	(	1.11	)	&	0.91	(	1.10	)	&	0.93	(	1.08	)	&	0.94	(	1.06	)	&	0.95	(	1.05	)	&	0.97	(	1.03	)	&	1	(	1	)	\\
0	&	1	(	1.00	)	&	1.00	(	1.00	)	&	1.00	(	1.00	)	&	1.00	(	1.00	)	&	1.00	(	1.00	)	&	1.00	(	1.00	)	&	1	(	1	)	\\
0.1	&	1.1	(	0.91	)	&	1.09	(	0.92	)	&	1.08	(	0.93	)	&	1.06	(	0.94	)	&	1.05	(	0.95	)	&	1.03	(	0.97	)	&	1	(	1	)	\\
0.5	&	1.5	(	0.67	)	&	1.45	(	0.69	)	&	1.38	(	0.73	)	&	1.30	(	0.77	)	&	1.25	(	0.80	)	&	1.15	(	0.87	)	&	1	(	1	)	\\
1	&	2	(	0.50	)	&	1.90	(	0.53	)	&	1.75	(	0.57	)	&	1.60	(	0.63	)	&	1.50	(	0.67	)	&	1.30	(	0.77	)	&	1	(	1	)	\\
\hline
	\end{tabular}}
	\label{TAB:A}
\end{table}

For a discrete model family as considered above, the minimum PSD estimating equation is then given by
\begin{eqnarray}
\sum_{x:r_n(x)\ne 0} \left[ f_{\theta}^{1+\alpha}(x) -  f_{\theta}^{B}(x) r_n^{A}(x) \right] u_{\theta}(x)
+  h A \sum_{x:r_n(x)= 0} f_{\theta}^{1+\alpha}(x) u_{\theta}(x) &=& 0 \nonumber \\
\mbox{or,   } ~~ \sum K_h(\delta(x))f_{\theta}^{1+\alpha}(x) u_{\theta}(x) = 0,
\label{EQ:5MPSDE_est_eqn} 
\end{eqnarray}
where $\delta(x)= \delta_n(x) = \frac{r_n(x)}{f_{\theta}(x)}-1$ and 
\begin{eqnarray}
K_h(\delta) &=& \left\{\begin{array}{lr}
\frac{(\delta+1)^A - 1}{A} & \mbox{ if } \delta \ne -1, \\
- h & \mbox{ if } \delta = -1.
\end{array}\right.
\end{eqnarray}
Note that the estimating equation for the minimum penalized $S$-divergence estimator 
has the same form as that of the $S$-divergence case \citep{Ghosh/etc:2013}, except that the continuous $K(\delta)$ 
has been now transformed to $K_h(\delta)$ that is discontinuous at the lower end-point $\delta=-1$. 
However, due to this different structure of the functions $K(\delta)$ and $K_h(\delta)$, 
the asymptotic properties of the minimum penalized $S$-divergence estimators (MPSDE) 
cannot be obtained directly from that of the MSDEs.
We will rigorously derive the asymptotics of all MPSDEs in the next section.

\section{Asymptotic Properties of the MPSDE under Discrete Models}
\label{SEC:5MPSDE_asymp_discrete}

Consider the set-up of discrete models as described above.
Note that, the estimating equations of MPSDE and MSDE only differ in terms of the functions
$K(\delta)$ and $K_h(\delta)$. 
We will follow the approach of \cite{Ghosh:2013} in establishing 
the asymptotic properties of the minimum divergence estimators, 
while clearly indicating the required modifications needed to extend the proof to the case of the penalized version.
Note that, intuitively this difference between the two estimating equations vanishes asymptotically 
under the true model, because the set of possible points $x$ with $r_n(x)=0$ should converge 
to a null set under the true distribution. Hence the asymptotic distribution of the
MPSDEs may intuitively be expected to be the same as that of the MSDES 
and there is no loss of asymptotic efficiency 
in using the penalized $S$-divergence over $S$-divergence.
The theorem below presents the same with more concrete proof, 
which extends the proof for the ordinary $S$-divergences 
\citep[Theorem 1,][]{Ghosh:2013} to this present case of penalized $S$-divergences.

Let us start with some useful Lemmas.
Along with the notations of \cite{Ghosh:2013}, 
consider the definitions $a_n(x) = K(\delta_n(x)) - K(\delta_g(x))$ and 
$b_n(x) = (\delta_n(x)-\delta_g(x))K'(\delta_g(x))$,
where $\delta_g(x) = \frac{g(x)}{f_\theta(x)}-1$.
Further assume that Conditions (SA1)--(SA7) of \cite{Ghosh:2013} hold. 
Then the following two lemmas help us to obtain the asymptotic distribution of \\ \begin{center}
$S_{1n}^* = \sqrt n \displaystyle\sum_{x:r_n(x)\neq 0} a_n(x)f_{\theta}^{1+\alpha}(x)u_{\theta}(x)$ 
and $S_{2n}^* = \sqrt n \displaystyle\sum_{x:r_n(x) \neq 0} b_n(x)f_{\theta}^{1+\alpha}(x)u_{\theta}(x)$.
\end{center}

\begin{lemma}
	Assume that Condition (SA5) holds. Then $E|S_{1n}^*-S_{2n}^*| \rightarrow 0 $ as $n \rightarrow \infty$, \\
and hence $S_{1n}-S_{2n} \displaystyle\mathop{\rightarrow}^\mathcal{P} 0$ as $n \rightarrow \infty$.
\label{LEM:5S1n_S2n_equiv}
\end{lemma}
\begin{proof}
Following the same line of the proof of the Lemma 3 in \cite{Ghosh:2013}, we get
\begin{eqnarray}
E|S_{1n}^*-S_{2n}^*| 
&\le&  \beta \sum_{x:r_n(x) \neq 0} g^{1/2}(x) f_{\theta}^{\alpha}(x)|u_{\theta}(x)| \nonumber \\
&\le&  \beta \sum_{x} g^{1/2}(x) f_{\theta}^{\alpha}(x)|u_{\theta}(x)| \nonumber \\
& < &   \infty, ~~~~~~ \mbox{ (by assumption (SA5))}.\nonumber
\end{eqnarray}
Then the proof follows using the dominated convergence theorem (DCT) and  Markov inequality.
\end{proof}

\begin{lemma}
Suppose the matrix $V_g$, as defined in Lemma 4 of \cite{Ghosh:2013}, is finite. Then 
$$
S_{1n}^* \mathop{\rightarrow}^\mathcal{D} N(0, V_g).
$$
\label{LEM:5S1n_distr}
\end{lemma}
\begin{proof}	
Note that, 
\begin{eqnarray}
S_{2n}^* &=& \sqrt n \sum_{x:r_n(x) \neq 0} (\delta_n(x) - \delta_g(x)) K'(\delta_g(x)) 
f_{\theta}^{1+\alpha}(x) u_{\theta}(x) \nonumber \\
&=&  \sqrt n \sum_{x:r_n(x) \neq 0} (r_n(x) - g(x)) K'(\delta_g(x)) f_{\theta}^{\alpha}(x) u_{\theta}(x) \nonumber\\
&& ~+ \sqrt{n}~ \sum_{x:r_n(x) = 0} g(x) K'(\delta_g(x)) f_{\theta}^{\alpha}(x) u_{\theta}(x)  \nonumber \\
&=&  \sqrt n \left( \frac{1}{n} \sum_{i=1}^n \left[ K'(\delta_g(X_i)) f_{\theta}^{\alpha}(X_i) u_{\theta}(X_i) 
- E_g\{ K'(\delta_g(X))f_{\theta}^{\alpha}(X)u_{\theta}(X)\} \right]\right) \nonumber\\
&& ~ + \sqrt{n}~ \sum_{x:r_n(x) = 0} g(x) K'(\delta_g(x)) f_{\theta}^{\alpha}(x) u_{\theta}(x). 
\label{EQ:S_2n_expression} 
\end{eqnarray}
Now, the first term in above converges in distribution to $N(0,V_g)$.
Considering the second term, we will show that
\begin{eqnarray}
\sqrt{n}~ \sum_{x:r_n(x) = 0} g(x) K'(\delta_g(x)) f_{\theta}^{\alpha}(x) u_{\theta}(x)
\mathop{\rightarrow}^\mathcal{P} 0.
\label{EQ:5eq-1}
\end{eqnarray}
Note that
\begin{eqnarray}
&& \sqrt{n}~ \sum_{x:r_n(x) = 0} g(x) K'(\delta_g(x)) f_{\theta}^{\alpha}(x) u_{\theta}(x)\nonumber\\
&=& \sqrt{n}~ \sum_{x} g(x) K'(\delta_g(x)) f_{\theta}^{\alpha}(x) u_{\theta}(x)I(r_n(x)),
\end{eqnarray}
where $I(y)=1$ if $y=0$ and $0$ otherwise. thus,
\begin{eqnarray}
&& E_g\left[\sqrt{n}~ 
\left|\sum_{x:r_n(x) = 0} g(x) K'(\delta_g(x)) f_{\theta}^{\alpha}(x) u_{\theta}(x)\right|\right] \nonumber\\
&=& E_g\left[\sqrt{n}~ 
\left|\sum_{x} g(x) K'(\delta_g(x)) f_{\theta}^{\alpha}(x) u_{\theta}(x) I(r_n(x))\right|\right] \nonumber\\
&\leq& \sum_{x} \left|g^{1/2}(x)K'(\delta_g(x)) f_{\theta}^{\alpha}(x) u_{\theta}(x)\right|  
\left[\sqrt{n}~ g^{1/2}(x)\left\{1-g(x)\right\}^n \right]\nonumber\\
&\leq& C_1 \sum_{x} g^{1/2}(x)f_{\theta}^{\alpha}(x) \left|u_{\theta}(x)\right|  
\left[\sqrt{n}~ g^{1/2}(x)\left\{1-g(x)\right\}^n \right]. \nonumber
\end{eqnarray}
The last inequality follows by  Assumption (SA7) and the strong law of large numbers (SLLN), 
under which 
\begin{eqnarray}
|K'(\delta)| = |(\delta+1)^{A-1}| < 2C = C_1,  ~~~~~~ \mbox{(say)}.
\label{EQ:K1_bound}
\end{eqnarray}
Further, for all $0<x<1$, $\left[\sqrt{n}~ g^{1/2}(x)\left\{1-g(x)\right\}^n \right] \rightarrow 0$
as $n\rightarrow\infty$ and its maximum over $0<x<1$ is bounded by $1/\sqrt{2}$. 
Hence, by assumption (SA5) and DCT it follows that 
\begin{eqnarray}
E_g\left[\sqrt{n}~ 
\left|\sum_{x:r_n(x) = 0} g(x) K'(\delta_g(x)) f_{\theta}^{\alpha}(x) u_{\theta}(x)\right|\right] 
\rightarrow 0.  \nonumber
\end{eqnarray}
Then, by Markov inequality, it follows that the second term in (\ref{EQ:S_2n_expression})
goes to zero in probability as $n\rightarrow\infty$ and so 
$S_{2n}^* \displaystyle\mathop{\rightarrow}^\mathcal{D} N(0, V_g)$.
Combining this with the previous Lemma, we have the required result.
\end{proof}

\begin{theorem}\label{THM:5discrete_asymp_MPSDE}
Under Assumptions (SA1)--(SA7) of \cite{Ghosh:2013},	
there exists a consistent sequence $\hat\theta_n$ of roots to the 
minimum penalized $S$-divergence  estimating equation (\ref{EQ:S-est_eqn_discrete}).
Also, the asymptotic distribution of $\sqrt n (\hat\theta_n - \theta^g)$ is $p$-dimensional
normal with mean $0$ and variance $J_g^{-1}V_g J_g^{-1}$,
where $J_g$ and $V_g$ are as defined in \cite{Ghosh:2013}.\\
\end{theorem}
\begin{proof} 
\textit{Consistency:} Following the proof of Theorem 1 of \cite{Ghosh:2013}, 
let us consider the behavior of $PSD^h_{(\alpha, \lambda)}(\mathbf{r}_n,\mathbf{f}_{\theta})$ on a sphere $Q_a$ 
of radius $a$ and center at $\theta^g$. We wish to show that, for sufficiently small $a$, 
\begin{equation}
PSD^h_{(\alpha, \lambda)}(\mathbf{r}_n,\mathbf{f}_{\theta}) > 
PSD^h_{(\alpha, \lambda)}(\mathbf{r}_n,\mathbf{f}_{\theta^g}) ~~ \forall \theta ~~ \mbox{ on the surface of } Q_a,
\label{EQ:5MPSDE_asymp_to_prove}
\end{equation}
with probability tending to one so that the penalized $S$-divergence also has 
a local minimum with respect to $\theta$ in the interior of $Q_a$. At a local 
minimum, the estimating equations must be satisfied. Therefore, for any $a>0$ sufficiently small, 
the minimum $S$-divergence  estimating equation have a solution $\theta_n$ within $Q_a$ with probability 
tending to one as $n \rightarrow \infty$.

Now taking a Taylor series expansion of $PSD^h_{(\alpha, \lambda)}(\mathbf{r}_n,\mathbf{f}_{\theta})$ 
about $\theta = \theta^g$, we get 
\begin{eqnarray}
&& PSD^h_{(\alpha, \lambda)}(\mathbf{r}_n,\mathbf{f}_{\theta^g}) - 
PSD^h_{(\alpha, \lambda)}(\mathbf{r}_n,\mathbf{f}_{\theta})
\nonumber \\\nonumber \\
&& ~~~~ = - \sum_j (\theta_j - \theta_j^g)\nabla_j 
PSD^h_{(\alpha, \lambda)}(\mathbf{r}_n,\mathbf{f}_{\theta})|_{\theta = \theta^g} 
\nonumber \\
&&  	~~~~~~~~~~~~~~~ - \frac{1}{2} 
\sum_{j,k}  (\theta_j - \theta_j^g)(\theta_k - \theta_k^g)\nabla_{jk} 
PSD^h_{(\alpha, \lambda)}(\mathbf{r}_n,\mathbf{f}_{\theta})|_{\theta = \theta^g} 
\nonumber \\
&&  	~~~~~~~~~~~~~~~ -      \frac{1}{6} \sum_{j,k,l} (\theta_j - \theta_j^g)(\theta_k - \theta_k^g)
(\theta_l-\theta_l^g)\nabla_{jkl} PSD^h_{(\alpha, \lambda)}(\mathbf{r}_n,\mathbf{f}_{\theta})|_{\theta = \theta^*} 
\nonumber \\\nonumber \\
&& ~~~~ = S_1 + S_2 + S_3, ~~~~~~~~ (say) \nonumber 
\end{eqnarray}
where $\theta^*$ lies between $\theta^g$ and $\theta$. 

For the linear term $S_1$, we consider
\begin{eqnarray}
\nabla_j PSD^h_{(\alpha, \lambda)}(\mathbf{r}_n,\mathbf{f}_{\theta})|_{\theta = \theta^g} 
&=& -(1+\alpha) ~ \sum_{x : r_n(x) \neq 0}  K(\delta_n^g(x))f_{\theta^g}^{1+\alpha}(x)u_{j\theta^g}(x) \nonumber\\
&& ~~~ + h (1+\alpha)  \sum_{x : r_n(x) = 0}  f_{\theta^g}^{1+\alpha}(x)u_{j\theta^g}(x)
\end{eqnarray}
where $\delta_n^g(x)$ is the $\delta_n(x)$ evaluated at $\theta = \theta^g$. We will now show that 
\begin{eqnarray}
\sum_{x : r_n(x) \neq 0}  K(\delta_n^g(x))f_{\theta^g}^{1+\alpha}(x)u_{j\theta^g}(x) 
\displaystyle\mathop{\rightarrow}^\mathcal{P} \sum_x K(\delta_g^g(x)) f_{\theta^g}^{1+\alpha}(x) u_{j\theta^g}(x), 
\end{eqnarray}
as $n \rightarrow \infty$ and note that the right hand side of above is zero by definition of 
the minimum PSD estimator. Note that the one-term Taylor series expansion yields
\begin{eqnarray}
& &   \left|\sum_{x : r_n(x) \neq 0}  K(\delta_n^g(x))f_{\theta^g}^{1+\alpha}(x)u_{j\theta^g}(x) 
- \sum_{x}   K(\delta_g^g(x)) f_{\theta^g}^{1+\alpha}(x) u_{j\theta^g}(x)\right|  \nonumber \\
&\le &   \left|\sum_{x : r_n(x) \neq 0}  K(\delta_n^g(x))f_{\theta^g}^{1+\alpha}(x)u_{j\theta^g}(x) 
- \sum_{x : r_n(x) \neq 0}  K(\delta_g^g(x)) f_{\theta^g}^{1+\alpha}(x) u_{j\theta^g}(x) \right|  \nonumber\\
&& ~~~ + \sum_{x : r_n(x) = 0} \left| K(\delta_g^g(x)) f_{\theta^g}^{1+\alpha}(x) u_{j\theta^g}(x) \right| 
\nonumber \\
&\le & C_1 \sum_{x : r_n(x) \neq 0}  |\delta_n^g(x) - \delta_g^g(x)| f_{\theta^g}^{1+\alpha}(x) |u_{j\theta^g}(x)| 
+ \sum_{x : r_n(x) = 0} \left| K(\delta_g^g(x)) f_{\theta^g}^{1+\alpha}(x) u_{j\theta^g}(x) \right|  \nonumber \\
&& ~~~~~~~~~~~~~~~~~~~~~~~~~~\mbox{ (by Equation \ref{EQ:K1_bound})}\nonumber\\
&\le & C_1 \sum_x  |\delta_n^g(x) - \delta_g^g(x)| f_{\theta^g}^{1+\alpha}(x) |u_{j\theta^g}(x) |  
+ \sum_{x : r_n(x) = 0} \left| K(\delta_g^g(x)) f_{\theta^g}^{1+\alpha}(x) u_{j\theta^g}(x) \right|. \nonumber 
\end{eqnarray}
But, it was proved in Theorem 1 of \cite{Ghosh:2013} that, as $n \rightarrow \infty$ 
$$
\displaystyle\sum_x |\delta_n^g(x)-\delta_g^g(x)| f_{\theta^g}^{1+\alpha}(x)|u_{j\theta^g}(x)| 
\mathop{\rightarrow}^\mathcal{P} 0.
$$
Further, along the lines of the proof of the convergence result in (\ref{EQ:5eq-1}) as given in Lemma \ref{LEM:5S1n_distr}, 
one can show that, as  $n \rightarrow \infty$ 
$$
\displaystyle\sum_{x : r_n(x)=0} |K(\delta_g^g(x)) f_{\theta^g}^{1+\alpha}(x) u_{j\theta^g}(x) |=o_p(n^{-1/2}),
\mbox{ and }
\displaystyle\sum_{x : r_n(x)=0}  f_{\theta^g}^{1+\alpha}(x)u_{j\theta^g}(x) =  o_p(n^{-1/2}).
$$ 
Combining these, we get $\nabla_j PSD^h_{(\alpha, \lambda)}(\mathbf{r}_n,\mathbf{f}_{\theta})|_{\theta = \theta^g} 
\mathop{\rightarrow}^\mathcal{P} 0$, as  $n \rightarrow \infty$. 
Thus, with probability tending to one, $|S_1| < p a^3$, 
where $p$ is the dimension of $\theta$ and $a$ is the radius of $Q_a$.
	

By a similar extension of the proof of Theorem 1 of \cite{Ghosh:2013},
we can show that there exists $c > 0$ and $a_0 > 0$ such that for $a < a_0$, we have  $ S_2 < -c a^2$ 
with probability tending to one and $|S_3|< b a^3$ on the sphere $Q_a$ with probability tending to one.
This implies that (\ref{EQ:5MPSDE_asymp_to_prove}) holds, completing the proof of  the consistency part.
\\

\noindent
\textit{Asymptotic Normality:}		For the asymptotic normality, 
let us rewrite the estimating equation of the MPSDE in (\ref{EQ:5MPSDE_est_eqn}) as
\begin{equation}
\sum_{x : r_n(x) \neq 0}  K(\delta_n^g(x))f_{\theta^g}^{1+\alpha}(x)u_{j\theta^g}(x) 
- h \sum_{x : r_n(x) = 0}  f_{\theta^g}^{1+\alpha}(x)u_{j\theta^g}(x) = 0.
\label{EQ:5MPSDE_est_asymp}
\end{equation}
Now the second term in the left-hand side (LHS) of Equation (\ref{EQ:5MPSDE_est_asymp}) 
converges to zero in probability as proved above in the consistency part.
We will expand the first term in the LHS of (\ref{EQ:5MPSDE_est_asymp}) 
in a Taylor series about $\theta = \theta^g$ to get
\begin{eqnarray}
& & \sum_{x : r_n(x) \neq 0} K(\delta_n(x)) f_{\theta}^{1+\alpha}(x)u_{\theta}(x)  \nonumber \\
&& ~~~ = \sum_{x : r_n(x) \neq 0} K(\delta_n^g(x)) f_{\theta^g}^{1+\alpha}(x)u_{\theta^g}(x) 
\nonumber \\
&& ~~~ + \sum_k (\theta_k - \theta_k^g) \nabla_k 
\left[ \sum_{x : r_n(x) \neq 0} K(\delta_n(x))f_{\theta}^{1+\alpha}(x)u_{\theta}(x) \right]_{\theta = \theta^g} 
\nonumber \\
&& ~~~ + \frac{1}{2} \sum_{k,l} (\theta_k - \theta_k^g)(\theta_l - \theta_l^g)\nabla_{kl} 
\left[\sum_{x : r_n(x) \neq 0} K(\delta_n(x)) f_{\theta}^{1+\alpha}(x)u_{\theta}(x)\right]_{\theta = \theta'},    
\label{EQ:5_59}
\end{eqnarray}
where $\theta'$ lies in between $\theta$ and $\theta^g$.
Now, let $\theta_n$ be the solution of the minimum PSD estimating equation, 
which exist and is consistent by the previous part. 
Replacing $\theta$ by $\theta_n$ in Equation (\ref{EQ:5_59}), its LHS becomes zero, yielding
\begin{eqnarray}
&& - \sqrt n \sum_{x : r_n(x) \neq 0} K(\delta_n^g(x)) f_{\theta^g}^{1+\alpha}(x)u_{\theta^g}(x)  \nonumber \\ 
&& ~~~ =  \sqrt n \sum_k (\theta_{nk} - \theta_k^g) \times \left\{ \nabla_k 
\left[\sum_{x : r_n(x)\neq 0}K(\delta_n(x))f_{\theta}^{1+\alpha}(x)u_{\theta}(x)\right]_{\theta=\theta^g}\right. 
\nonumber \\
&& ~~~ + \left. \frac{1}{2} \sum_{l} (\theta_{nl} - \theta_l^g)\nabla_{kl} 
\left[\sum_{x : r_n(x)\neq 0} K(\delta_n(x))f_{\theta}^{1+\alpha}(x)u_{\theta}(x)\right]_{\theta=\theta'}\right\}. 
\label{EQ:5_60}
\end{eqnarray}
Note that, the first term within the bracketed quantity in the RHS of Equation (\ref{EQ:5_60}) 
converges to $J_g$ with probability tending to one (by a proof similar to that in Lemma \ref{LEM:5S1n_distr}), 
while the second bracketed term is an $o_p(1)$ term (as proved in the proof of consistency part). 
Also, by using Lemma \ref{LEM:5S1n_distr}, we get
\begin{eqnarray}
&& \sqrt n \sum_{x : r_n(x) \neq 0} K(\delta_n^g(x)) f_{\theta^g}^{1+\alpha}(x)u_{\theta^g}(x) \nonumber \\
&&  ~~~  =  \sqrt n \sum_{x : r_n(x) \neq 0}  [K(\delta_n^g(x)) - K(\delta_g^g(x))]  
f_{\theta^g}^{1+\alpha}(x)u_{\theta^g}(x) \nonumber \\
&& ~~~ =  S_{1n}^*|_{\theta=\theta^g} \mathop{\rightarrow}^\mathcal{D} N_p(0, V_g). 
\end{eqnarray}
Therefore, the theorem follows by the Lemma 4.1 of \cite{Lehmann:1983}.
\end{proof} 

\bigskip
In the particular case, when the true distribution $G$ belongs to the model family with $G = F_\theta$ 
for some $\theta \in \Theta$, then $\theta^g=\theta$ and $\sqrt n (\theta_n - \theta)$ 
has asymptotic distribution as $N_p(0, J^{-1} V J^{-1})$, where $J=M_{\alpha}$ and 
$V=M_{2\alpha} - N_\alpha N_\alpha^T$ with 
\begin{eqnarray}
M_\alpha =  \int u_\theta(x)u_\theta^T(x)f_\theta^{1+\alpha}(x) dx,
~~~~~~~~~~
N_\alpha = \int u_\theta(x)f_\theta^{1+\alpha}(x)dx.\nonumber
\end{eqnarray}
Note that, as in the case of the minimum $S$-divergence estimators, 
their penalized version also has asymptotic distribution independent of the parameter $\lambda$ 
under the model family.

Therefore we have observed that the first order asymptotic properties of the 
minimum penalized $S$-divergence estimator at the model family 
is exactly the same as that of the minimum $S$-divergence estimator.
This implies that the first order influence function of these two estimators will also be the same
so that the robustness properties of these two estimators are expected to be equivalent. 
Therefore, the minimum penalized $S$-divergence estimators generalize 
the minimum $S$-divergence estimators with no loss in their asymptotic efficiency and 
no degradation in their  robustness properties, 
and provide us the extra facility of inlier correction
at small sample sizes. Intuitively, it is quite clear that the 
performance of the penalized divergences in terms of their ability to successfully handle 
the inliers and empty cells in a small sample would depend on the choice of the penalty factor $h$.
In the next section, we will examine this characteristic of the 
minimum penalized $S$-divergence estimators through an extensive simulation study under 
the Poisson model with small sample sizes.

\section{Numerical Illustrations : Choice of the Penalty Factor $h$}
\label{SEC:choice_h}

In the previous section we have seen that the asymptotic properties of the MPSDEs are the same as those 
of the MSDEs. However, due to the special nature of their construction, the small sample properties of the MPSDEs are often
significantly different; we expect that the MPSDEs will have substantially superior performance 
in the presence of inliers and empty cells provided the penalty factor $h$, 
which has a crucial role in determining the small sample performance of the MPSDE, is  chosen carefully. 
In this section, we will empirically examine these small sample performances 
of the proposed penalized $S$-divergence estimators under the Poisson model family.

For our numerical illustrations, we generate random samples of small sizes ($n = 10, 15, 20$) 
from a Poisson distribution with mean $\theta$ and 
compute the minimum penalized $S$-divergence estimator of the parameter $\theta$ under the Poisson model for each given sample.
The sample generation process is replicated  $1000$ times to generate the empirical MSE of the MPSDEs; 
we then compare the performance of the MPSDEs based on the empirical MSE 
for different values of tuning parameters $\alpha$, $\lambda$  and $h$. 
We will restrict ourselves only to the range of positive $h$, which eliminates the possibility of a negative divergence. 
Further, since the small sample properties usually depend on the mean of the Poisson model, 
we will consider several values of $\theta$ ranging from $3$ to $9$.
For brevity in presentation, we will only present some interesting cases in 
Figures \ref{FIG:5MSE_MPSDE_n10} to \ref{FIG:5MSE_MPSDE_n20}.

\begin{figure}
\centering
\subfloat[$\alpha=0$, $\lambda=0$]{
\includegraphics[width=0.23\textwidth]{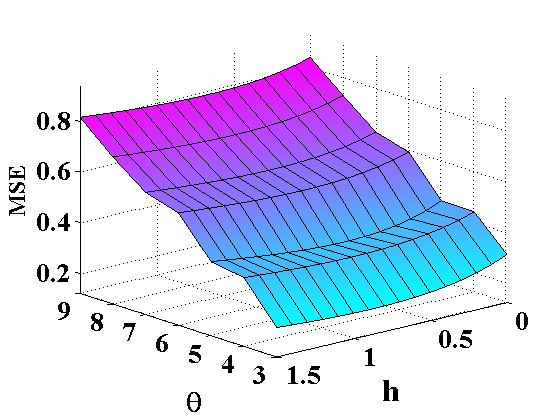}
\label{fig:MPSDE_n10_L0_a0}}
~ 
\subfloat[$\alpha=0.1$, $\lambda=0$]{
\includegraphics[width=0.23\textwidth]{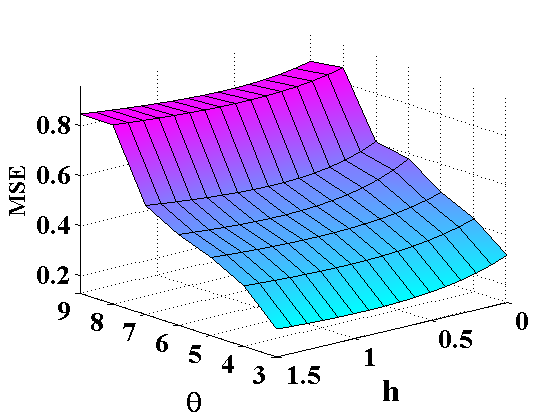}
\label{fig:MPSDE_n10_L0_a1}}
~ 
\subfloat[$\alpha=0.25$, $\lambda=0$]{
\includegraphics[width=0.23\textwidth]{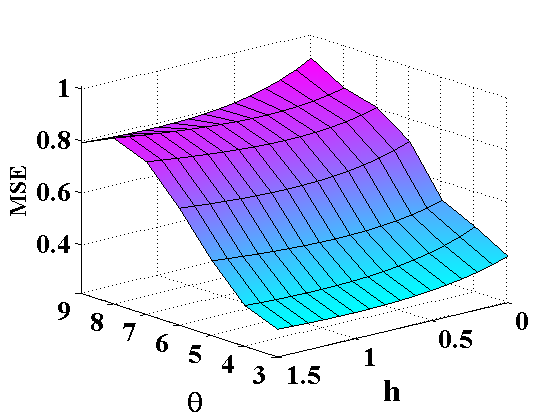}
\label{fig:MPSDE_n10_L0_a25}}
~ 
\subfloat[$\alpha=0.5$, $\lambda=0$]{
\includegraphics[width=0.23\textwidth]{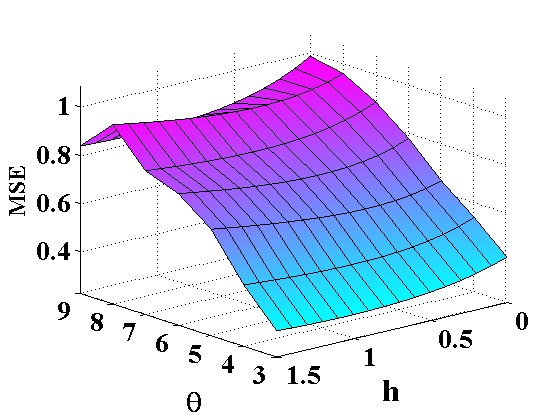}
\label{fig:MPSDE_n10_L0_a5}}
\\ 
\subfloat[$\alpha=0$, $\lambda=-0.5$]{
\includegraphics[width=0.23\textwidth]{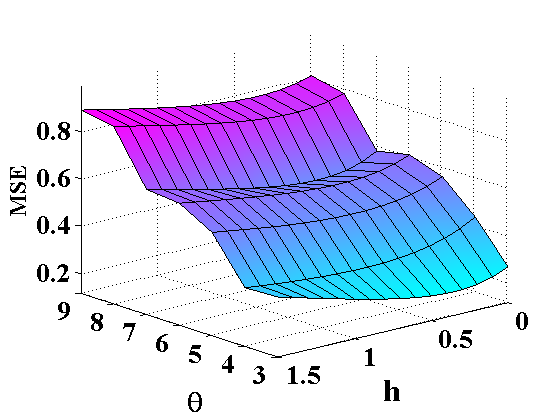}
\label{fig:MPSDE_n10_L-5_a0}}
~ 
\subfloat[$\alpha=0.1$, $\lambda=-0.5$]{
\includegraphics[width=0.23\textwidth]{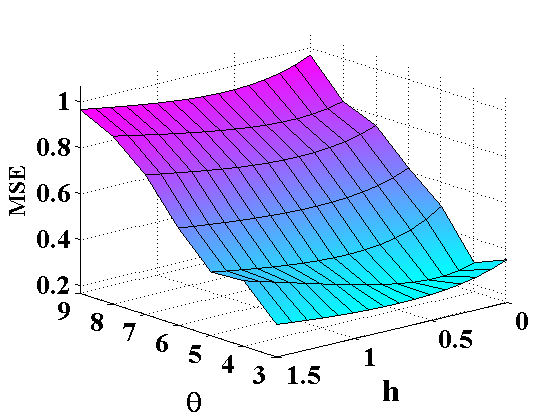}
\label{fig:MPSDE_n10_L-5_a1}}
~ 
\subfloat[$\alpha=0.25$, $\lambda=-0.5$]{
\includegraphics[width=0.23\textwidth]{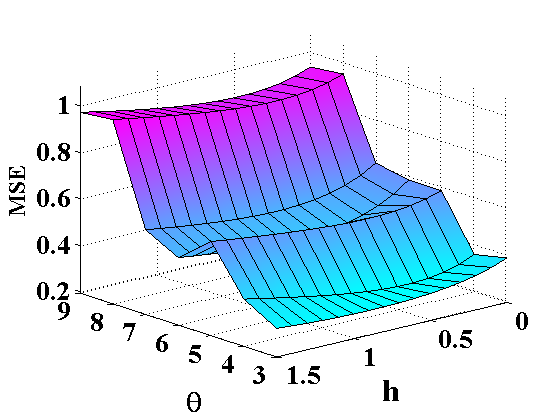}
\label{fig:MPSDE_n10_L-5_a25}}
~ 
\subfloat[$\alpha=0.5$, $\lambda=-0.5$]{
\includegraphics[width=0.23\textwidth]{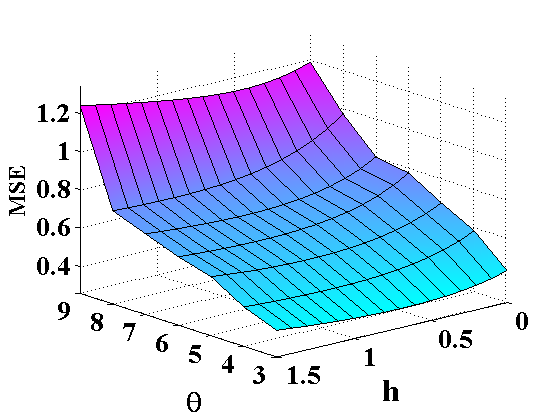}
\label{fig:MPSDE_n10_L-5_a5}}
\\ 
\subfloat[$\alpha=0$, $\lambda=-1$]{
\includegraphics[width=0.23\textwidth]{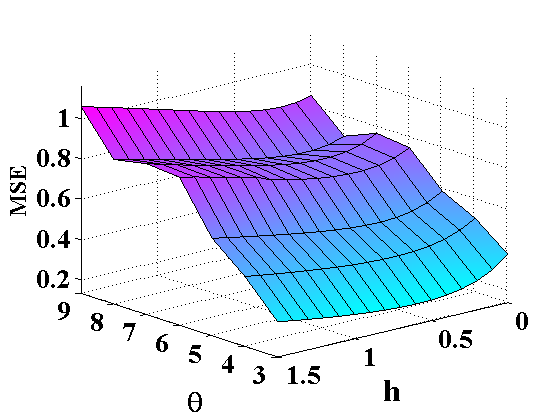}
\label{fig:MPSDE_n10_L-10_a0}}
~ 
\subfloat[$\alpha=0.1$, $\lambda=-1$]{
\includegraphics[width=0.23\textwidth]{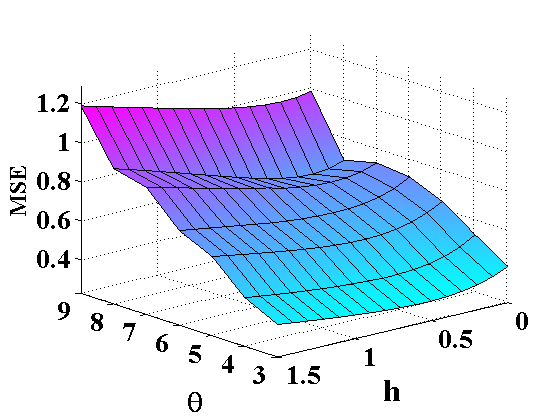}
\label{fig:MPSDE_n10_L-10_a1}}
~ 
\subfloat[$\alpha=0.25$, $\lambda=-1$]{
\includegraphics[width=0.23\textwidth]{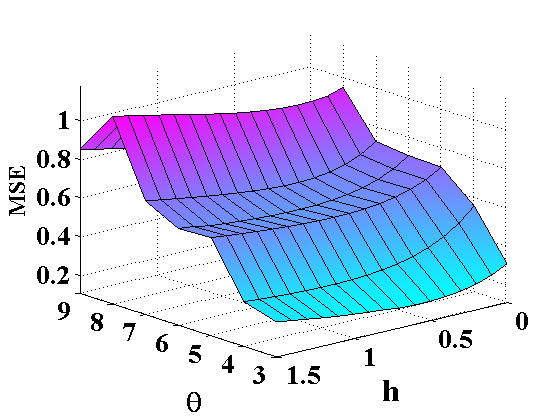}
\label{fig:MPSDE_n10_L-10_a25}}
~ 
\subfloat[$\alpha=0.5$, $\lambda=-1$]{
\includegraphics[width=0.23\textwidth]{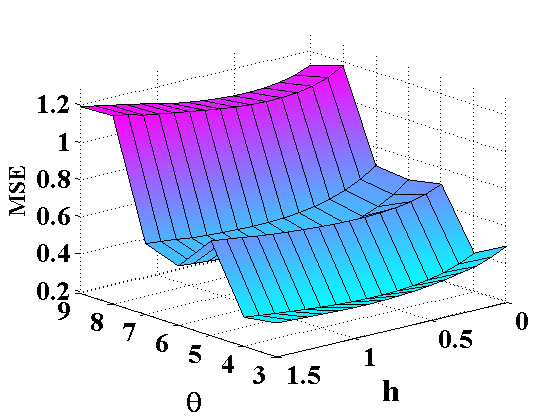}
\label{fig:MPSDE_n10_L-10_a5}}
\\
\subfloat[$\alpha=0$, $\lambda=-1.5$]{
\includegraphics[width=0.23\textwidth]{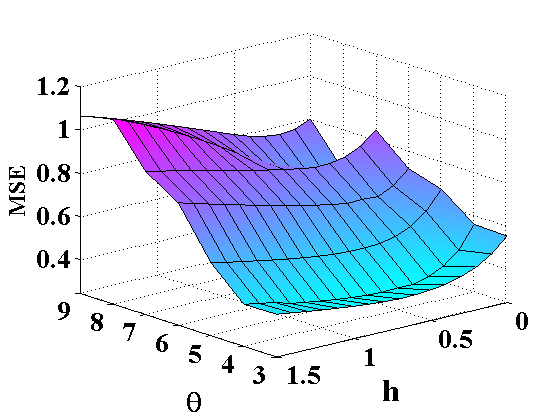}
\label{fig:MPSDE_n10_L-15_a0}}
~ 
\subfloat[$\alpha=0.1$, $\lambda=-1.5$]{
\includegraphics[width=0.23\textwidth]{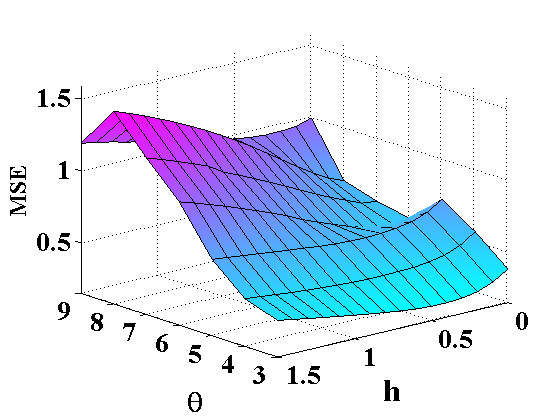}
\label{fig:MPSDE_n10_L-15_a1}}
~ 
\subfloat[$\alpha=0.25$, $\lambda=-1.5$]{
\includegraphics[width=0.23\textwidth]{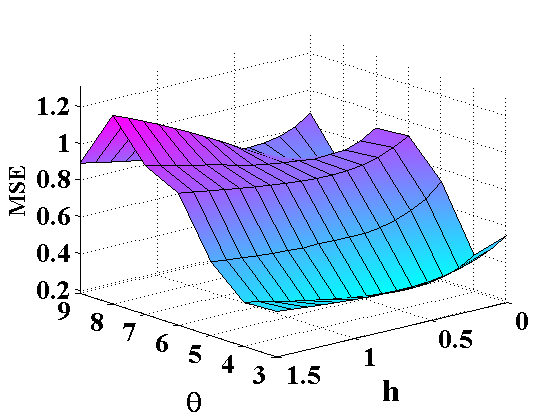}
\label{fig:MPSDE_n10_L-15_a25}}
~ 
\subfloat[$\alpha=0.5$, $\lambda=-1.5$]{
\includegraphics[width=0.23\textwidth]{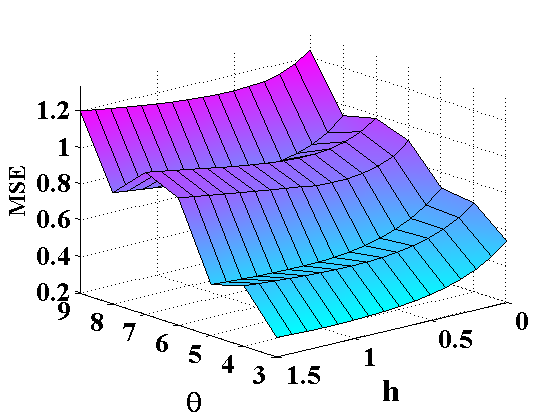}
\label{fig:MPSDE_n10_L-15_a5}}
\\ 
\subfloat[$\alpha=0$, $\lambda=-2$]{
\includegraphics[width=0.23\textwidth]{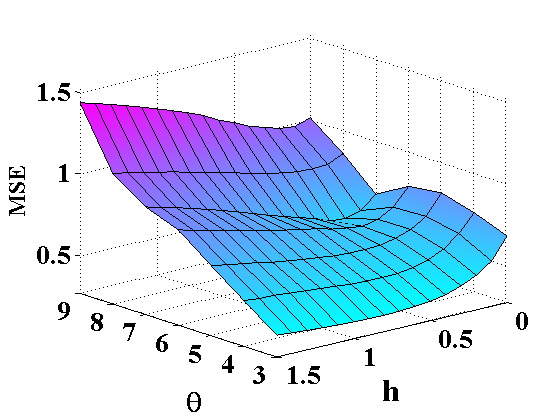}
\label{fig:MPSDE_n10_L-20_a0}}
~ 
\subfloat[$\alpha=0.1$, $\lambda=-2$]{
\includegraphics[width=0.23\textwidth]{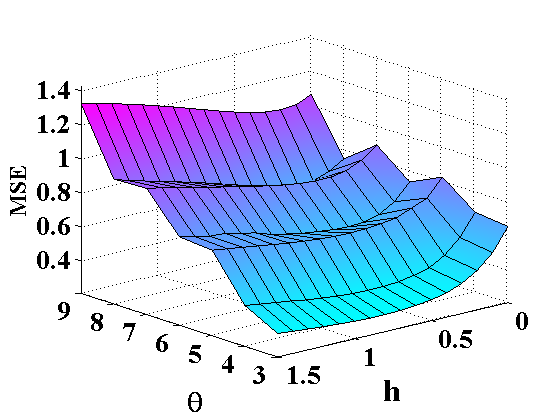}
\label{fig:MPSDE_n10_L-20_a1}}
~ 
\subfloat[$\alpha=0.25$, $\lambda=-2$]{
\includegraphics[width=0.23\textwidth]{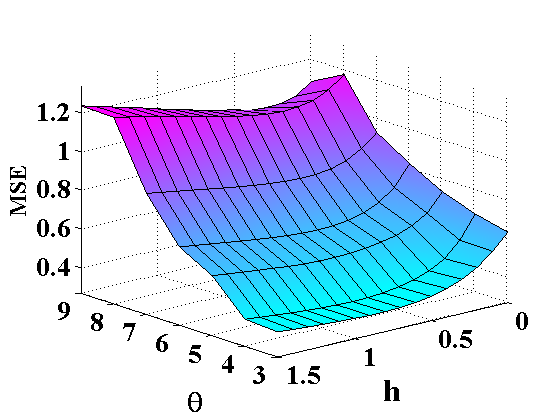}
\label{fig:MPSDE_n10_L-20_a25}}
~ 
\subfloat[$\alpha=0.5$, $\lambda=-2$]{
\includegraphics[width=0.23\textwidth]{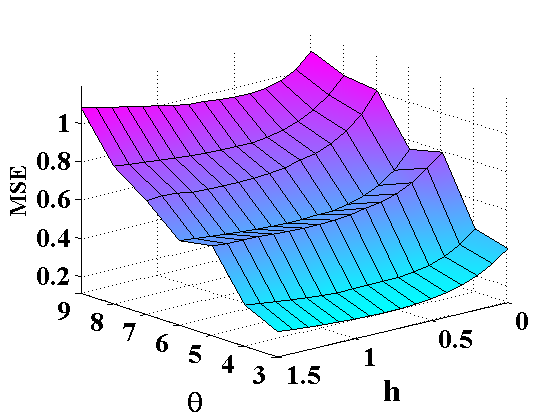}
\label{fig:MPSDE_n10_L-20_a5}}
\caption{MSE of $\widehat{\theta}_{\alpha, \lambda}$ over $\theta$ and $h$ for different $\lambda$ and $\alpha$ at sample size $n=10$. }
\label{FIG:5MSE_MPSDE_n10}
\end{figure}

\begin{figure}
\centering
\subfloat[$\alpha=0$, $\lambda=0$]{
\includegraphics[width=0.23\textwidth]{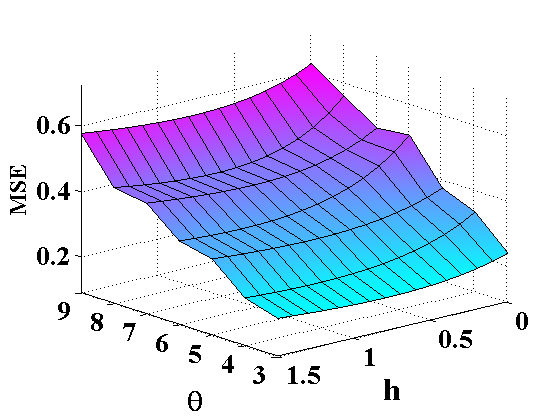}
\label{fig:MPSDE_n15_L0_a0}}
~ 
\subfloat[$\alpha=0.1$, $\lambda=0$]{
\includegraphics[width=0.23\textwidth]{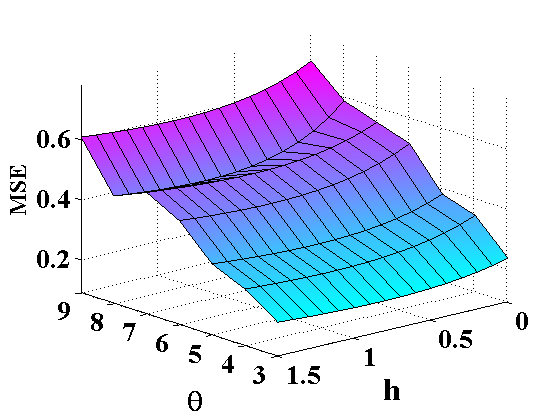}
\label{fig:MPSDE_n15_L0_a1}}
~ 
\subfloat[$\alpha=0.25$, $\lambda=0$]{
\includegraphics[width=0.23\textwidth]{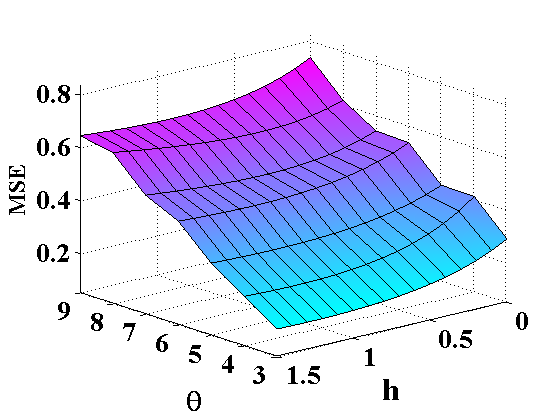}
\label{fig:MPSDE_n15_L0_a25}}
~ 
\subfloat[$\alpha=0.5$, $\lambda=0$]{
\includegraphics[width=0.23\textwidth]{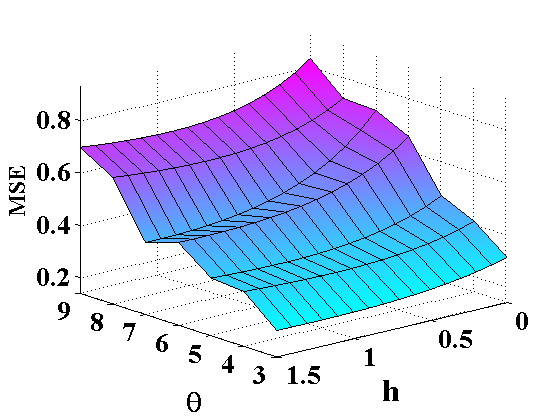}
\label{fig:MPSDE_n15_L0_a5}}
\\ 
\subfloat[$\alpha=0$, $\lambda=-0.5$]{
\includegraphics[width=0.23\textwidth]{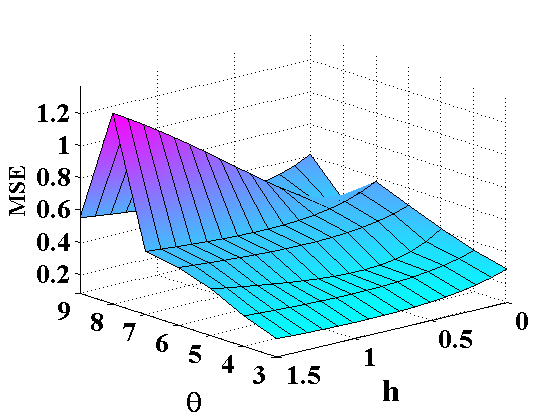}
\label{fig:MPSDE_n15_L-5_a0}}
~ 
\subfloat[$\alpha=0.1$, $\lambda=-0.5$]{
\includegraphics[width=0.23\textwidth]{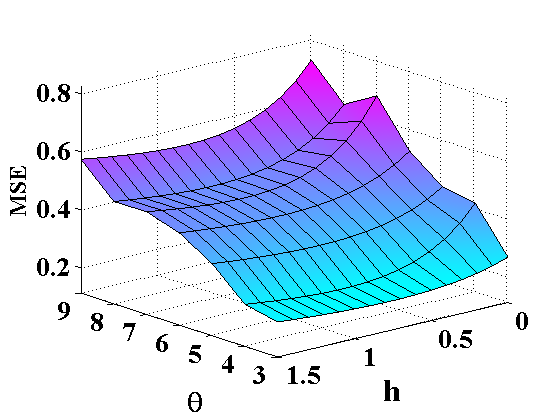}
\label{fig:MPSDE_n15_L-5_a1}}
~ 
\subfloat[$\alpha=0.25$, $\lambda=-0.5$]{
\includegraphics[width=0.23\textwidth]{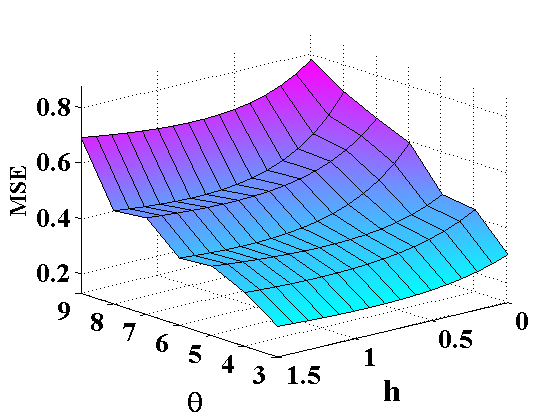}
\label{fig:MPSDE_n15_L-5_a25}}
~ 
\subfloat[$\alpha=0.5$, $\lambda=-0.5$]{
\includegraphics[width=0.23\textwidth]{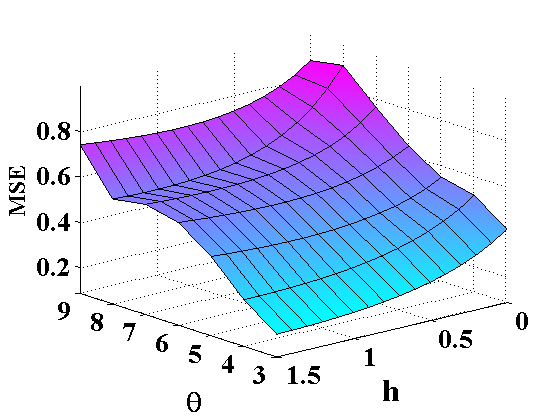}
\label{fig:MPSDE_n15_L-5_a5}}
\\ 
\subfloat[$\alpha=0$, $\lambda=-1$]{
\includegraphics[width=0.23\textwidth]{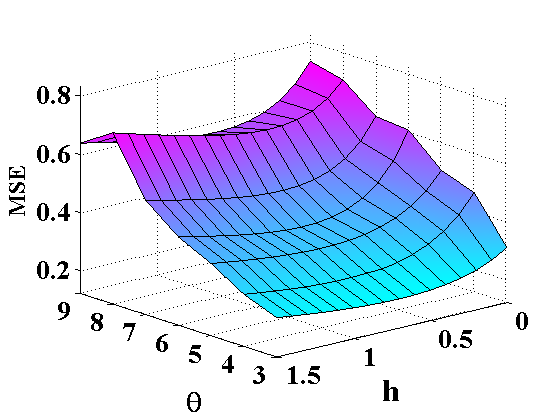}
\label{fig:MPSDE_n15_L-10_a0}}
~ 
\subfloat[$\alpha=0.1$, $\lambda=-1$]{
\includegraphics[width=0.23\textwidth]{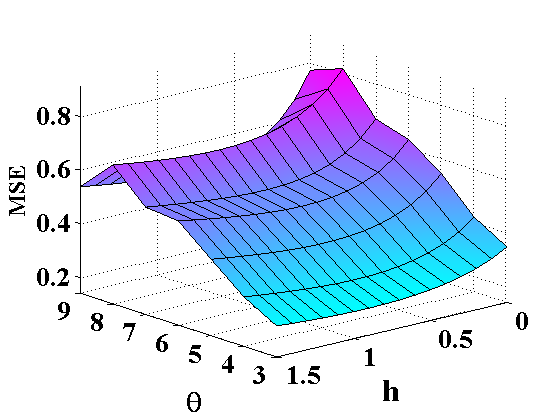}
\label{fig:MPSDE_n15_L-10_a1}}
~ 
\subfloat[$\alpha=0.25$, $\lambda=-1$]{
\includegraphics[width=0.23\textwidth]{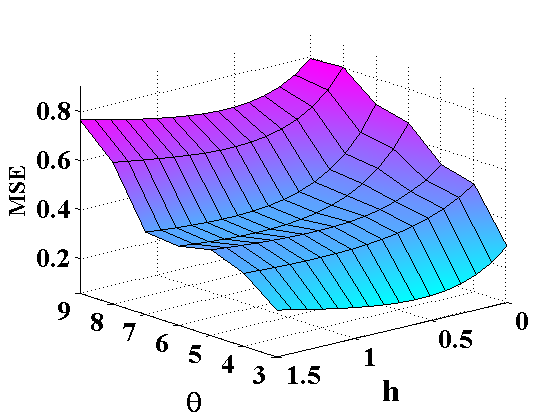}
\label{fig:MPSDE_n15_L-10_a25}}
~ 
\subfloat[$\alpha=0.5$, $\lambda=-1$]{
\includegraphics[width=0.23\textwidth]{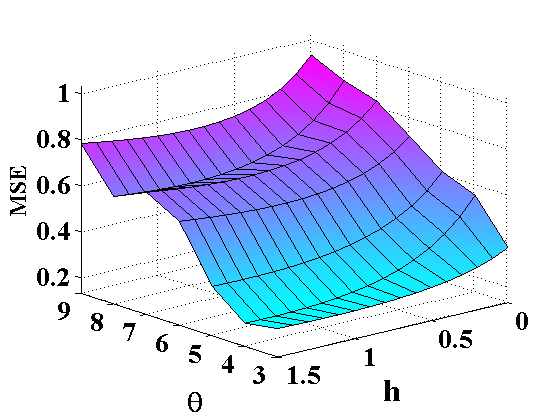}
\label{fig:MPSDE_n15_L-10_a5}}
\\
\subfloat[$\alpha=0$, $\lambda=-1.5$]{
\includegraphics[width=0.23\textwidth]{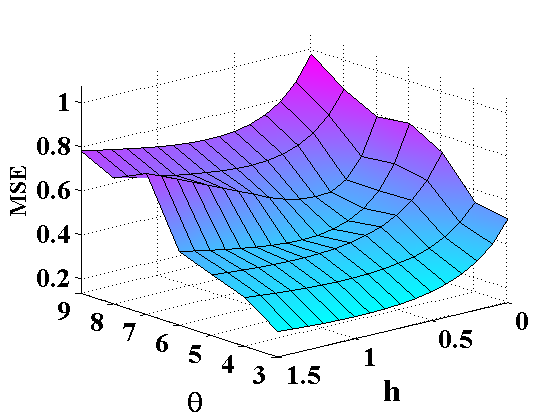}
\label{fig:MPSDE_n15_L-15_a0}}
~ 
\subfloat[$\alpha=0.1$, $\lambda=-1.5$]{
\includegraphics[width=0.23\textwidth]{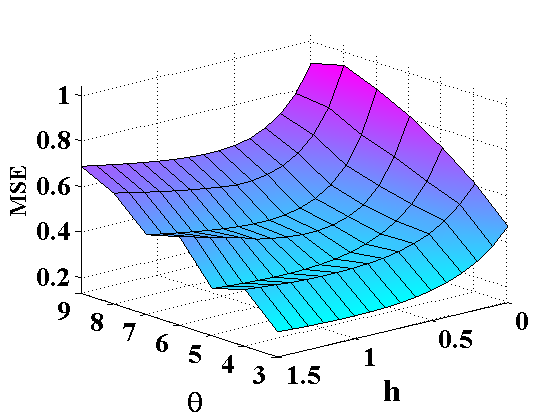}
\label{fig:MPSDE_n15_L-15_a1}}
~ 
\subfloat[$\alpha=0.25$, $\lambda=-1.5$]{
\includegraphics[width=0.23\textwidth]{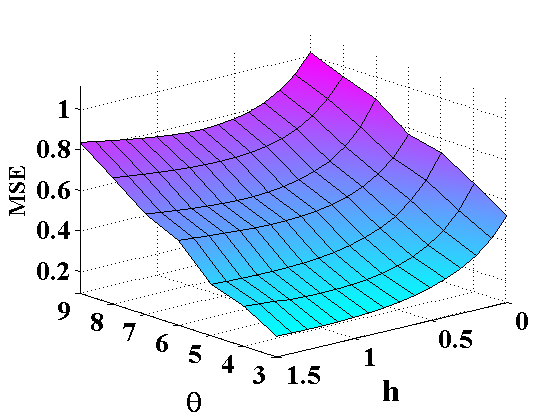}
\label{fig:MPSDE_n15_L-15_a25}}
~ 
\subfloat[$\alpha=0.5$, $\lambda=-1.5$]{
\includegraphics[width=0.23\textwidth]{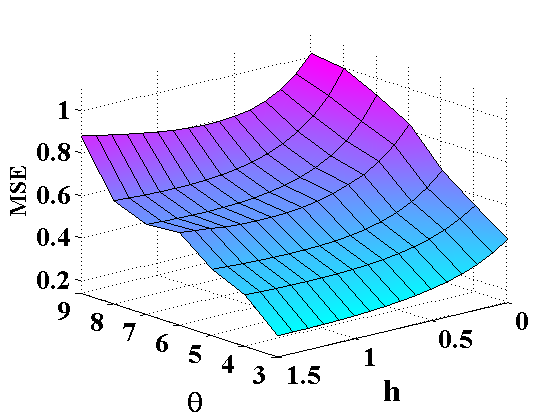}
\label{fig:MPSDE_n15_L-15_a5}}
\\ 
\subfloat[$\alpha=0$, $\lambda=-2$]{
\includegraphics[width=0.23\textwidth]{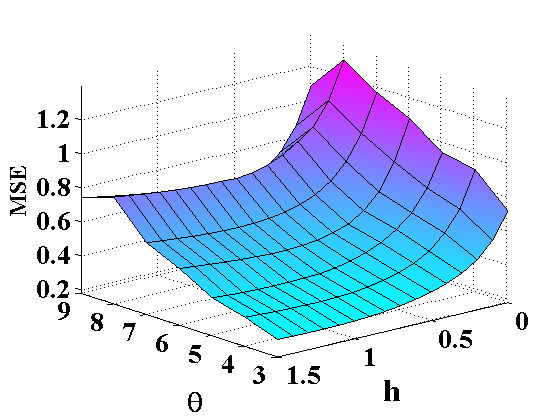}
\label{fig:MPSDE_n15_L-20_a0}}
~ 
\subfloat[$\alpha=0.1$, $\lambda=-2$]{
\includegraphics[width=0.23\textwidth]{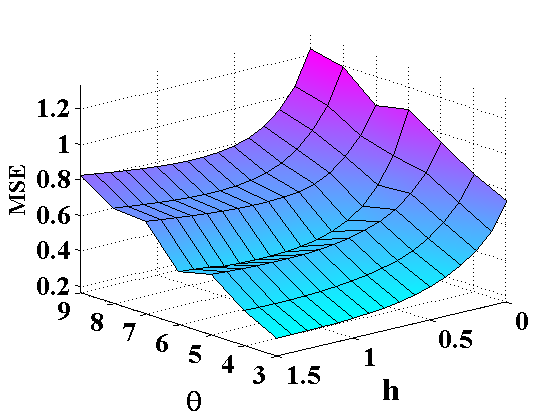}
\label{fig:MPSDE_n15_L-20_a1}}
~ 
\subfloat[$\alpha=0.25$, $\lambda=-2$]{
\includegraphics[width=0.23\textwidth]{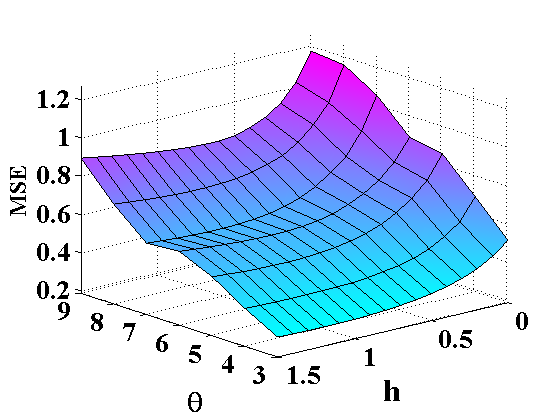}
\label{fig:MPSDE_n15_L-20_a25}}
~ 
\subfloat[$\alpha=0.5$, $\lambda=-2$]{
\includegraphics[width=0.23\textwidth]{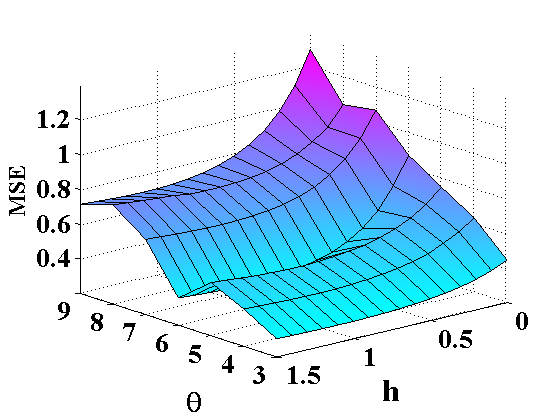}
\label{fig:MPSDE_n15_L-20_a5}}
\caption{MSE of $\widehat{\theta}_{\alpha, \lambda}$ over $\theta$ and $h$ for different $\lambda$ and $\alpha$ at sample size $n=15$. }
\label{FIG:5MSE_MPSDE_n15}
\end{figure}

\begin{figure}
\centering
\subfloat[$\alpha=0$, $\lambda=0$]{
\includegraphics[width=0.23\textwidth]{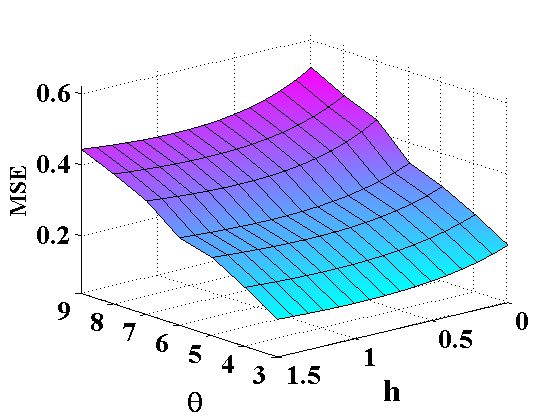}
\label{fig:MPSDE_n20_L0_a0}}
~ 
\subfloat[$\alpha=0.1$, $\lambda=0$]{
\includegraphics[width=0.23\textwidth]{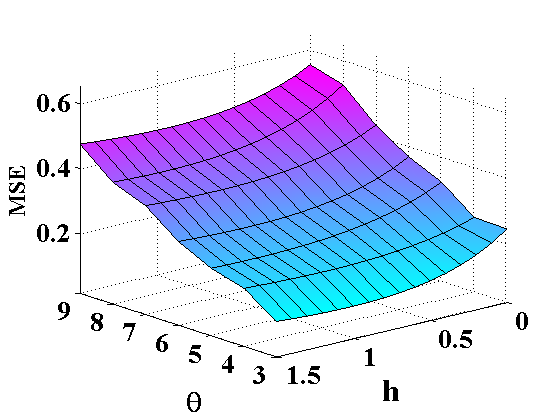}
\label{fig:MPSDE_n20_L0_a1}}
~ 
\subfloat[$\alpha=0.25$, $\lambda=0$]{
\includegraphics[width=0.23\textwidth]{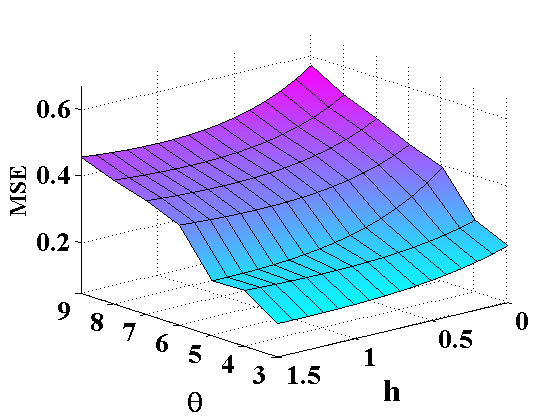}
\label{fig:MPSDE_n20_L0_a25}}
~ 
\subfloat[$\alpha=0.5$, $\lambda=0$]{
\includegraphics[width=0.23\textwidth]{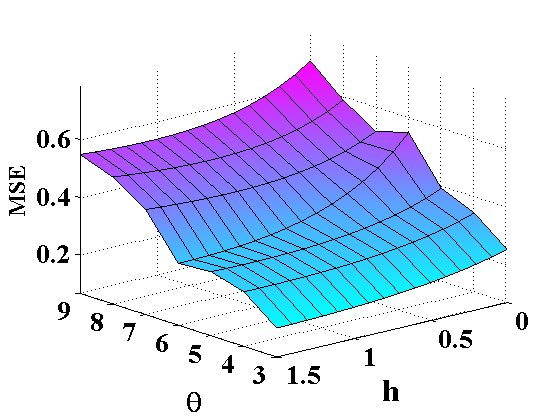}
\label{fig:MPSDE_n20_L0_a5}}
\\ 
\subfloat[$\alpha=0$, $\lambda=-0.5$]{
\includegraphics[width=0.23\textwidth]{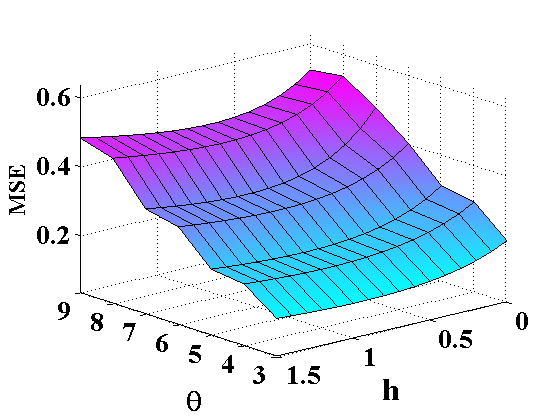}
\label{fig:MPSDE_n20_L-5_a0}}
~ 
\subfloat[$\alpha=0.1$, $\lambda=-0.5$]{
\includegraphics[width=0.23\textwidth]{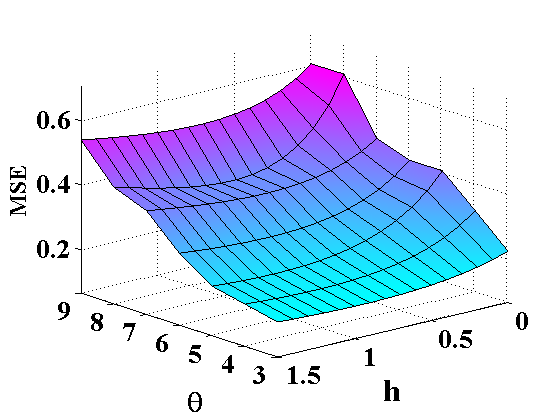}
\label{fig:MPSDE_n20_L-5_a1}}
~ 
\subfloat[$\alpha=0.25$, $\lambda=-0.5$]{
\includegraphics[width=0.23\textwidth]{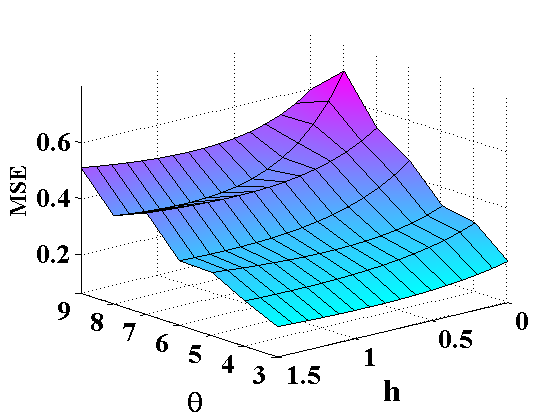}
\label{fig:MPSDE_n20_L-5_a25}}
~ 
\subfloat[$\alpha=0.5$, $\lambda=-0.5$]{
\includegraphics[width=0.23\textwidth]{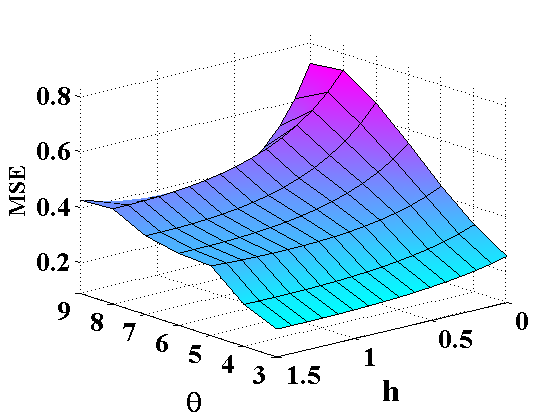}
\label{fig:MPSDE_n20_L-5_a5}}
\\ 
\subfloat[$\alpha=0$, $\lambda=-1$]{
\includegraphics[width=0.23\textwidth]{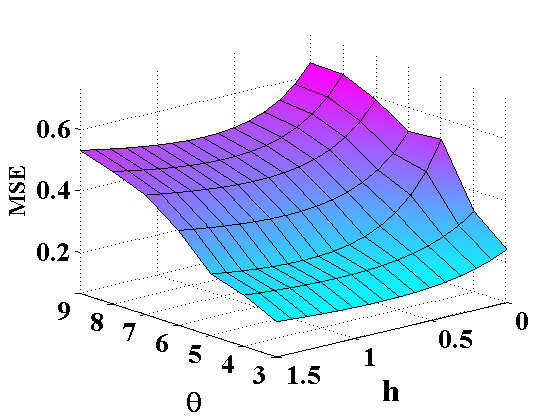}
\label{fig:MPSDE_n20_L-10_a0}}
~ 
\subfloat[$\alpha=0.1$, $\lambda=-1$]{
\includegraphics[width=0.23\textwidth]{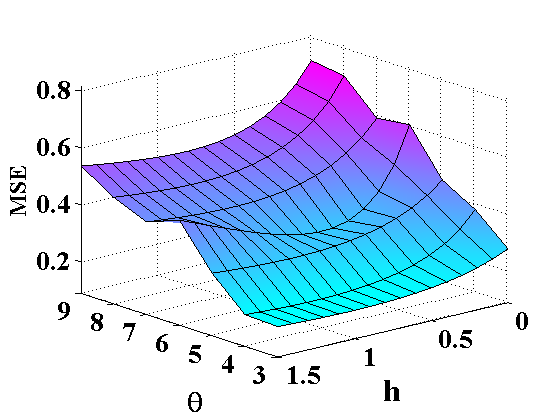}
\label{fig:MPSDE_n20_L-10_a1}}
~ 
\subfloat[$\alpha=0.25$, $\lambda=-1$]{
\includegraphics[width=0.23\textwidth]{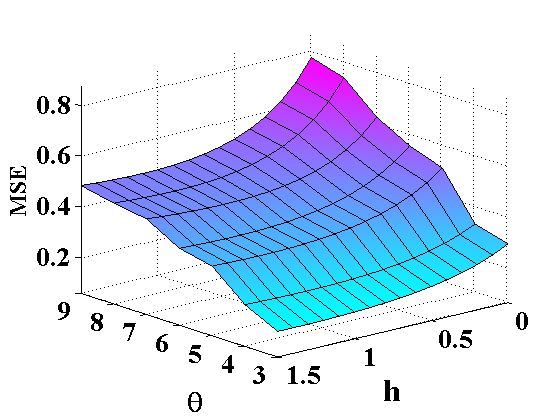}
\label{fig:MPSDE_n20_L-10_a25}}
~ 
\subfloat[$\alpha=0.5$, $\lambda=-1$]{
\includegraphics[width=0.23\textwidth]{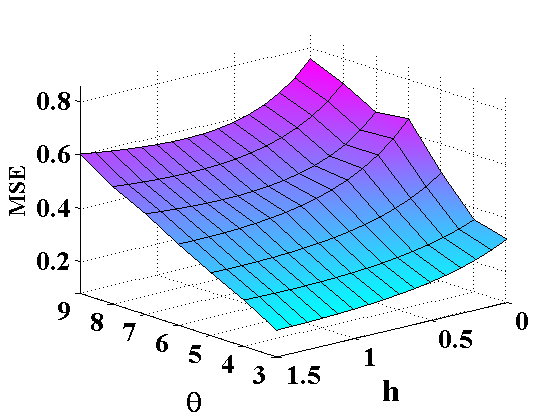}
\label{fig:MPSDE_n20_L-10_a5}}
\\
\subfloat[$\alpha=0$, $\lambda=-1.5$]{
\includegraphics[width=0.23\textwidth]{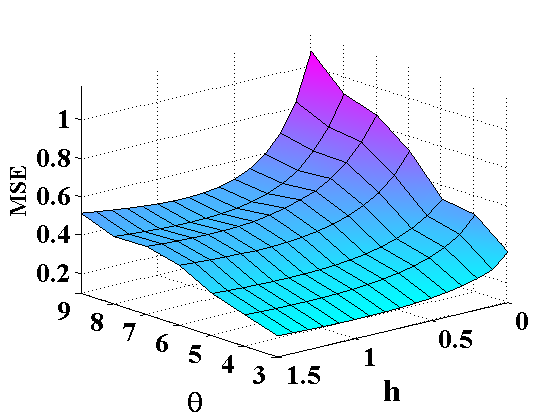}
\label{fig:MPSDE_n20_L-15_a0}}
~ 
\subfloat[$\alpha=0.1$, $\lambda=-1.5$]{
\includegraphics[width=0.23\textwidth]{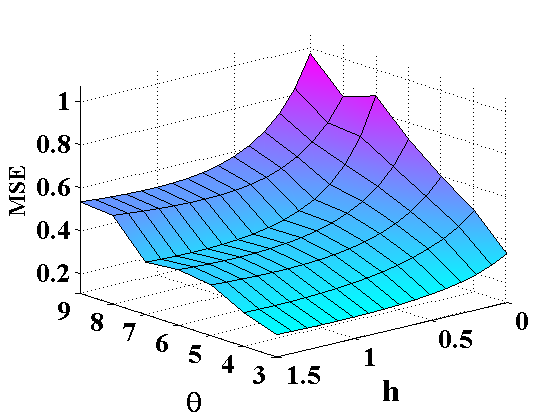}
\label{fig:MPSDE_n20_L-15_a1}}
~ 
\subfloat[$\alpha=0.25$, $\lambda=-1.5$]{
\includegraphics[width=0.23\textwidth]{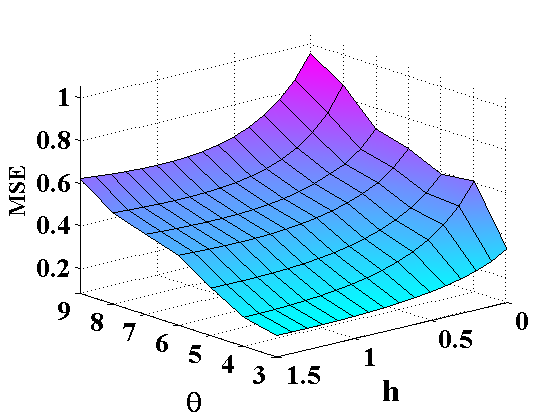}
\label{fig:MPSDE_n20_L-15_a25}}
~ 
\subfloat[$\alpha=0.5$, $\lambda=-1.5$]{
\includegraphics[width=0.23\textwidth]{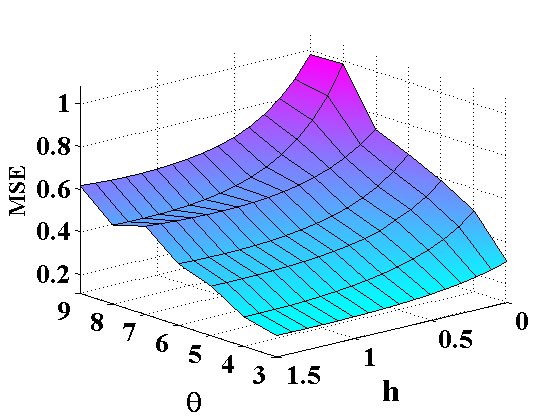}
\label{fig:MPSDE_n20_L-15_a5}}
\\ 
\subfloat[$\alpha=0$, $\lambda=-2$]{
\includegraphics[width=0.23\textwidth]{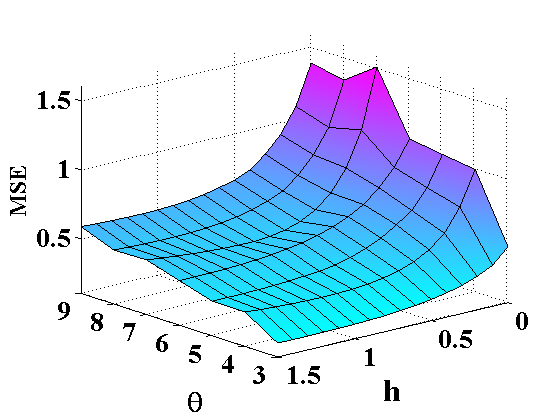}
\label{fig:MPSDE_n20_L-20_a0}}
~ 
\subfloat[$\alpha=0.1$, $\lambda=-2$]{
\includegraphics[width=0.23\textwidth]{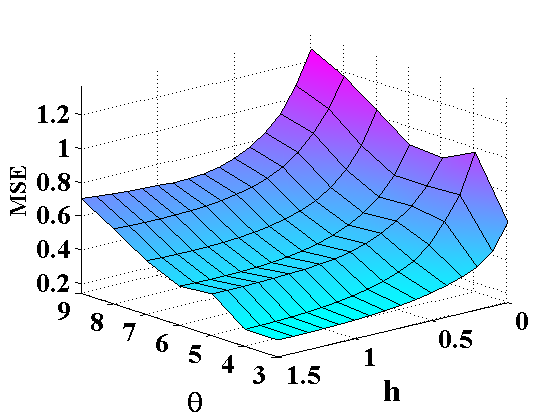}
\label{fig:MPSDE_n20_L-20_a1}}
~ 
\subfloat[$\alpha=0.25$, $\lambda=-2$]{
\includegraphics[width=0.23\textwidth]{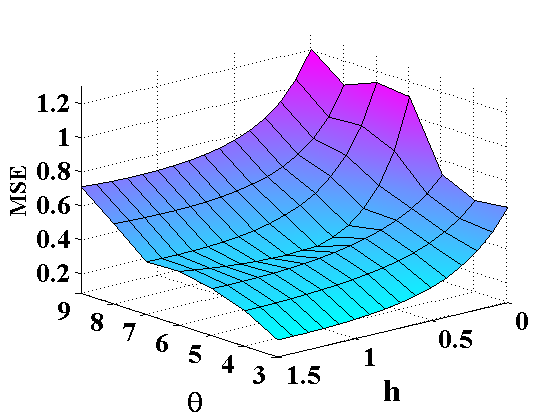}
\label{fig:MPSDE_n20_L-20_a25}}
~ 
\subfloat[$\alpha=0.5$, $\lambda=-2$]{
\includegraphics[width=0.23\textwidth]{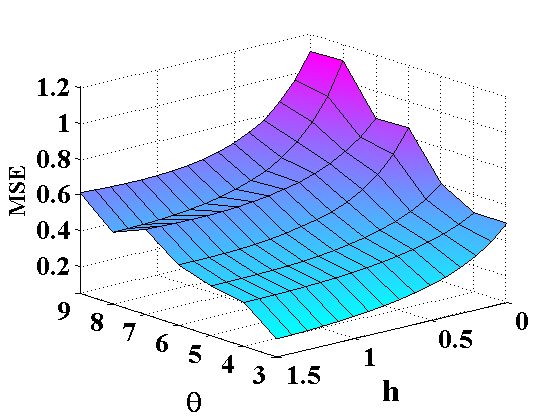}
\label{fig:MPSDE_n20_L-20_a5}}
\caption{MSE of $\widehat{\theta}_{\alpha, \lambda}$ over $\theta$ and $h$ for different $\lambda$ and $\alpha$ at sample size $n=20$. }
\label{FIG:5MSE_MPSDE_n20}
\end{figure}

It is easy to note from the figures that the pattern of the MSE over different values of 
the penalty factor $h$ and Poisson mean $\theta$ generally have a similar nature 
for all the sample sizes $n=10, 15$ and $20$.
In all the cases, the MSE increases with increasing $\theta$ and 
decreases slightly as the sample size increases.
However its behavior for different values of $\lambda$ and $\alpha$ varies significantly
and so does the optimum range for the values of the penalty factor $h$ generating 
the minimum MSE for different members of the $S$-divergence family.
Further, although it is not discernible from these graphs themselves,  
the minimum value of the MSE of the MPSDE (over the different choices of of the penalty factor $h$) 
for any ($\alpha$, $\lambda$) combination is generally much smaller compared to that of the corresponding MSDE. 
This illustrates that a suitably chosen penalized version of the MSDE can largely solve the problem of inliers 
and empty cells in small sample sizes.
And in almost all the situations studied, the optimal penalty factor is in the range $[0, 1.5]$. 
This is the reason why we have restricted our attention to this range for $h$. 
And in the overwhelming majority of the cases, the MSE surface shows very little change over 
$h \in [0.5, 1.5]$ for fixed values of $\alpha$, $\lambda$ and the Poisson parameter. 

\begin{table}[!th]
\centering 
\caption{Optimum values of the penalty factor $h$ for sample size $n=10$ 
and different values of $(\alpha, \lambda)$, along with the true empty cell weight = $\frac{1}{A}$}
\begin{tabular}{lr|c|ccccccc} \hline
$\alpha$	&	$\lambda$	&  $1/A$	& $\theta=3$	&	$\theta=4$	&	$\theta=5$	&	$\theta=6$	&	$\theta=7$	&	$\theta=8$	&	$\theta=9$ \\\hline 
0		&	0		&	1 		& 	0.9	&	0.8	&	0.6	&	0.7 &	0.7	&	0.7	&	0.7			\\
0.1		&	0		&	1 		& 	0.9	&	0.9	&	0.8	&	0.8 &	0.7	& 	0.7 &	0.6			\\
0.25	&	0		&	1		& 	1 	&	1	&	0.7	&	0.8	&	0.6	&	0.5	&	1			\\
0.5		&	0		&	1		&	1	&	0.8	&	0.5	&	0.6	&	0.8	&	0.6	&	1.1	 		\\\hline 
0		&	$-0.5$	&	2 		& 	0.4	&	0.7	&	0.6	&	0.5	&	0.4	&	0.3	&	0.4			\\
0.1		&	$-0.5$	&	1.82	& 	0.7	&	0.3	&	0.6	&	0.6	&	0.5	&	0.3	&	0.6			\\
0.25	&	$-0.5$	&	1.60	&	0.9	&	0.4	&	0.5	&	0.8	&	0.7	&	0.5	&	0.5			\\
0.5		&	$-0.5$	&	1.33	&	0.7	&	1	&	0.9	&	1	&	0.9	&	1.1	&	0.4			\\\hline 
0		&	$-1$	& $\infty$	&	0.6	&	0.5	&	0.4	&	0.4	&	0.4	&	0.4	&	0.4			\\
0.1		&	$-1$	&	10 		& 	0.5	&	0.6	&	0.5	&	0.5	&	0.3	&	0.2	&	0.2			\\
0.25	&	$-1$	&	4 		& 	0.6	&	0.9	&	0.7	&	0.7	&	0.5	&	0.3	&	0.4			\\
0.5		&	$-1$	&	2 		& 	0.9	&	0.7	&	0.6	&	0.9	&	0.8	&	0.5	&	0.4			\\\hline 
0		&	$-1.5$	&	$-2$ 	& 	0.7	&	0.6	&	0.5	&	0.2 &	0.4	&	0.1	&	0.2			\\
0.1		&	$-1.5$	&	$-2.86$	& 	0.4	&	0.6	&	0.5	&	0.1 &	0	&	0	&	0.3			\\
0.25	&	$-1.5$	&	$-8$ 	& 	0.7	&	0.4	&	0.9	&	0.5 &	0.4	&	0	&	0.4			\\
0.5		&	$-1.5$	&	$4$ 	& 	1.2	&	0.8	&	1	&	0.5	&	0.5	&	0.7	&	0.5			\\\hline 
0		&	$-2$	&	$-1$ 	& 	0.8	&	0.7	&	0.5	&	0.3 &	0.2	&	0.2	&	0.1			\\
0.1		&	$-2$	&	$-1.25$	& 	0.8	&	0.7	&	0.6	&	0.4 &	0.4	&	0.3	&	0.2			\\
0.25	&	$-2$	&	$-2$ 	& 	0.9	&	0.8	&	0.8	&	0.6 &	0.5	&	0.4	&	0.3			\\
0.5		&	$-2$	& $-\infty$	& 	0.8	&	0.9	&	0.9	&	0.8	&	1	&	0.6	&	0.4			\\
\hline
\end{tabular}
\label{TAB:opt_h_n10}
\end{table}

\begin{table}[!th]
	\centering 
	\caption{Optimum values of the penalty factor $h$ for sample size $n=15$ 
		and different values of $(\alpha, \lambda)$, along with the true empty cell weight = $\frac{1}{A}$}
	\begin{tabular}{lr|c|ccccccc} \hline
		$\alpha$	&	$\lambda$	&  $1/A$	& $\theta=3$	&	$\theta=4$	&	$\theta=5$	&	$\theta=6$	&	$\theta=7$	&	$\theta=8$	&	$\theta=9$	\\\hline 
		0		&	0		&	1 		& 	0.8	&	1	&	0.9	&	1.1 &	0.8	&	0.8	&	0.8			\\
		0.1		&	0		&	1 		& 	0.8	&	0.9	&	0.9	&	0.9 &	0.8	& 	0.9 &	0.9			\\
		0.25	&	0		&	1		& 	1.3 &	1.2	&	1.1	&	1	&	0.7	&	0.6	&	0.9			\\
		0.5		&	0		&	1		&	1.2	&	0.9	&	1.1	&	1.4	&	1.5	&	0.9	&	0.9	 		\\\hline 
		0		&	$-0.5$	&	2 		& 	0.9	&	0.8	&	0.6	&	0.7	&	0.7	&	0	&	0.6			\\
		0.1		&	$-0.5$	&	1.82	& 	0.8	&	0.9	&	0.8	&	0.8	&	0.9	&	0.9	&	0.8			\\
		0.25	&	$-0.5$	&	1.60	&	0.9	&	1	&	0.8	&	1.1	&	1	&	1	&	0.8			\\
		0.5		&	$-0.5$	&	1.33	&	1.5	&	1.5	&	1	&	0.9	&	1.1	&	1.3	&	0.9			\\\hline 
		0		&	$-1$	& $\infty$	&	0.7	&	0.9	&	0.8	&	0.7	&	0.6	&	0.5	&	0.6			\\
		0.1		&	$-1$	&	10 		& 	0.9	&	0.8	&	0.8	&	0.7	&	0.7	&	0.8	&	0.8			\\
		0.25	&	$-1$	&	4 		& 	0.6	&	1	&	0.8	&	1	&	1	&	0.7	&	0.6			\\
		0.5		&	$-1$	&	2 		& 	1.2	&	1.4	&	1.3	&	0.9	&	0.9	&	1	&	0.8			\\\hline 
		0		&	$-1.5$	&	$-2$ 	& 	1	&	0.9	&	0.9	&	0.7 &	0.4	&	0.6	&	0.7			\\
		0.1		&	$-1.5$	&	$-2.86$	& 	1	&	1	&	1	&	0.7 &	0.8	&	0.7	&	0.7			\\
		0.25	&	$-1.5$	&	$-8$ 	& 	1.3	&	1	&	1.1	&	0.9 &	0.9	&	0.8	&	0.7			\\
		0.5		&	$-1.5$	&	$4$ 	& 	1.4	&	1.1	&	1	&	1	&	1.1	&	1.1	&	0.7			\\\hline 
		0		&	$-2$	&	$-1$ 	& 	1.3	&	1	&	1.1	&	0.9 &	0.9	&	0.7	&	0.6			\\
		0.1		&	$-2$	&	$-1.25$	& 	1.5	&	1.3 &	1.1	&	1 	&	0.8	&	0.9	&	0.8			\\
		0.25	&	$-2$	&	$-2$ 	& 	1.2	&	1.2	&	1	&	0.9 &	1.1	&	1	&	0.6			\\
		0.5		&	$-2$	& $-\infty$	& 	1.1	&	1	&	1	&	1.1	&	1	&	0.9	&	1.1			\\
		\hline
	\end{tabular}
	\label{TAB:opt_h_n15}
\end{table}

\begin{table}[!th]
	\centering 
	\caption{Optimum values of the penalty factor $h$ for sample size $n=20$ 
		and different values of $(\alpha, \lambda)$, along with the true empty cell weight = $\frac{1}{A}$}
	\begin{tabular}{lr|c|ccccccc} \hline
		$\alpha$	&	$\lambda$	&  $1/A$	& $\theta=3$	&	$\theta=4$	&	$\theta=5$	&	$\theta=6$	&	$\theta=7$	&	$\theta=8$	&	$\theta=9$	 \\\hline 
		0		&	0		&	1 		& 	1	&	1	&	0.9	&	0.9 &	0.9	&	0.9	&	0.9			\\
		0.1		&	0		&	1 		& 	1	&	1	&	1.2	&	1.1 &	1	& 	1.2 &	0.9			\\
		0.25	&	0		&	1		& 	1.4 &	1	&	1.4	&	1	&	1	&	1	&	1			\\
		0.5		&	0		&	1		&	1.5	&	1.4	&	1.3	&	1.5	&	1	&	0.9	&	1.1	 		\\\hline 
		0		&	$-0.5$	&	2 		& 	1	&	1	&	0.9	&	0.9	&	0.9	&	0.8	&	0.7			\\
		0.1		&	$-0.5$	&	1.82	& 	1	&	1	&	1.1	&	0.9	&	0.7	&	0.7	&	0.7			\\
		0.25	&	$-0.5$	&	1.60	&	0.9	&	1.2	&	0.8	&	1.2	&	0.9	&	1	&	0.8			\\
		0.5		&	$-0.5$	&	1.33	&	1.4	&	1.3	&	1.1	&	1.4	&	1.5	&	1.5	&	0.8			\\\hline 
		0		&	$-1$	& $\infty$	&	0.9	&	0.9	&	0.9	&	0.8	&	0.8	&	0.8	&	0.7			\\
		0.1		&	$-1$	&	10 		& 	1	&	1.1	&	1	&	0.6	&	0.9	&	1	&	0.8			\\
		0.25	&	$-1$	&	4 		& 	1.4	&	1.1	&	1.2	&	1.2	&	1	&	1.1	&	1.1			\\
		0.5		&	$-1$	&	2 		& 	1.5	&	1.3	&	1.4	&	1.4	&	1	&	1	&	0.8			\\\hline 
		0		&	$-1.5$	&	$-2$ 	& 	1	&	1.1	&	1	&	1 	&	1.1	&	1	&	0.8			\\
		0.1		&	$-1.5$	&	$-2.86$	& 	1.1	&	1.3	&	1.1	&	1.1 &	1.2	&	0.9	&	1			\\
		0.25	&	$-1.5$	&	$-8$ 	& 	1.4	&	1.5	&	1.3	&	1.1 &	1.1	&	1.1	&	1			\\
		0.5		&	$-1.5$	&	$4$ 	& 	1.3	&	1.5	&	1.4	&	1.3	&	1	&	1.3	&	1.2			\\\hline 
		0		&	$-2$	&	$-1$ 	& 	1.4	&	1	&	1.4	&	1.5 &	0.9	&	1.2	&	1.1			\\
		0.1		&	$-2$	&	$-1.25$	& 	1.4	&	1.5	&	1.5	&	1.3 &	1.2	&	1.2	&	0.9			\\
		0.25	&	$-2$	&	$-2$ 	& 	1.5	&	1.5	&	1.5	&	0.8 &	1.3	&	1.2	&	1.2			\\
		0.5		&	$-2$	& $-\infty$	& 	1.5	&	1.5	&	1.5	&	1.4	&	1.2	&	1.5	&	1.2			\\
		\hline
	\end{tabular}
	\label{TAB:opt_h_n20}
\end{table}

From the figures it is clear that  the optimal choice of the penalty factor $h$ 
depends on the tuning parameters $(\alpha, \lambda)$, and possibly also on the sample size and the Poisson parameter $\theta$. 
In Tables \ref{TAB:opt_h_n10}--\ref{TAB:opt_h_n20}, we have presented the optimum values of the 
penalty factor $h$ for different $\theta$, $\alpha$ and  $\lambda$ for $n=10$, $15$ and $20$, 
along with the true empty cell weight $1/A$. 
In the following we give a structured description of what we observe in the tables. 
\begin{itemize}
	
	\item For the maximum likelihood estimator (corresponding to $\alpha = 0$, $\lambda = 0$), 
	the impact of the penalty factor is not substantial, and the optimal $h$ is usually close 
	to the natural factor $1/A$ (equal to 1 in this case). 
	
	\item The findings are similar for the $(\alpha > 0, \lambda = 0)$ cases, 
	although larger $\alpha$ and $n$ usually require slightly higher values of the penalty factor. 
	
	\item For the Hellinger distance $(\alpha = 0, \lambda = -0.5)$ there is significant improvement 
	in the MSE due to the imposition of the penalty. The optimum $h$ varies between 0.5 and 0.1 (the natural value is 2). 
	Although we have not presented the corresponding values in our tables or figures, 
	the improvement is even more spectacular for larger (in absolute magnitude) negative values of $\lambda$. 
	
	\item A pattern similar to the one described in the last item is observed for $\alpha > 0$ and $\lambda = -0.5$, 
	although the range of optimum values shifts slightly to larger values of $h$ as $n$ increases.

	\item In general, smaller sample sizes require greater shrinking of the penalty factor towards zero. 
	Very small sample sizes (such as 10), require penalty factors closer to 0.5, 
	while with moderate to larger values a factor closer to 1 (or slightly higher) is more appropriate. 
	
	\item In practically all the cases that we have investigated, 
	the optimal penalty factor is smaller (in absolute magnitude) than the natural penalty weight. 
	
	\item As the true mean parameter increases, the optimal penalty factor needed is, in general, smaller. 
	This is because in this case there are a larger number of possibly inlying cells with non-negligible probabilities. 
\end{itemize}

Thus, within the limitation of the simulation scheme, we have given a detailed description of 
how the penalty weight appears to improve the performance of these estimators. 
The optimal choice does depend on the parameter $\theta$, the sample size $n$ and the tuning parameter $(\alpha, \lambda)$. 
A completely automatic, case specific recommendation for the penalty factor to choose 
in a given situation may require additional research. However, as an overall recommendation, 
we observe that in general $h \in [0.5, 1]$ works well in practically all situations, 
with the specific choice being guided, to the extent possible, by the discussion above.

\begin{figure}
	\centering
	\subfloat[$n=10$, $\lambda=0$, $h^*=0.5$]{
		\includegraphics[width=0.23\textwidth]{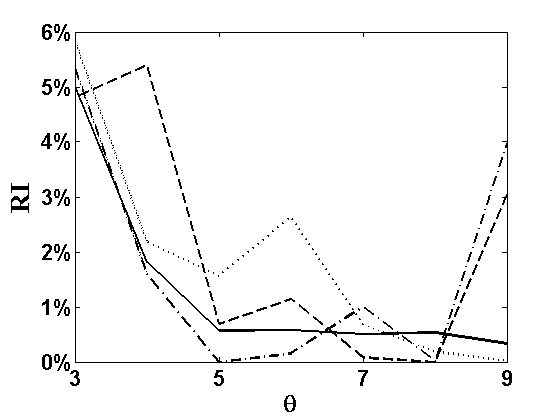}
		\label{fig:RI_n10_h05_l0}}
	~ 
	\subfloat[$n=10$, $\lambda=0$, $h^*=1$]{
		\includegraphics[width=0.23\textwidth]{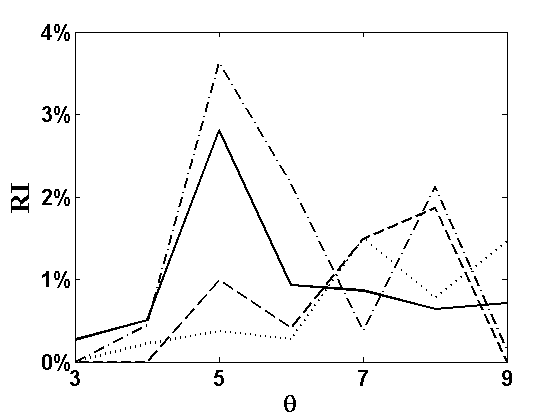}
		\label{fig:RI_n10_h10_l0}}
	~ 
	\subfloat[$n=20$, $\lambda=0$, $h^*=0.5$]{
	\includegraphics[width=0.23\textwidth]{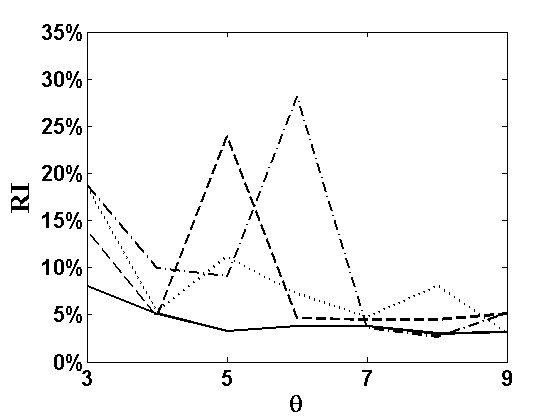}
	\label{fig:RI_n20_h05_l0}}
	~ 
	\subfloat[$n=20$, $\lambda=0$, $h^*=1$]{
	\includegraphics[width=0.23\textwidth]{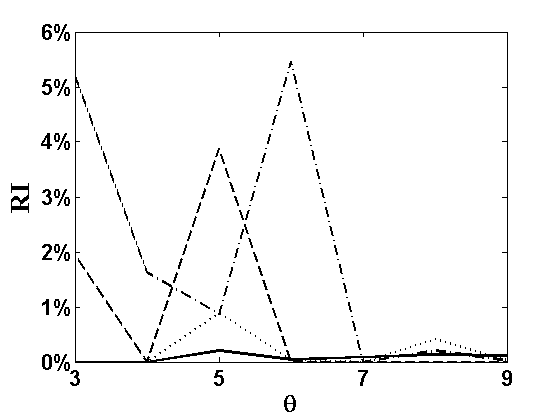}
	\label{fig:RI_n20_h10_l0}}
	\\ 
\subfloat[$n=10$, $\lambda=-0.5$, $h^*=0.5$]{
	\includegraphics[width=0.23\textwidth]{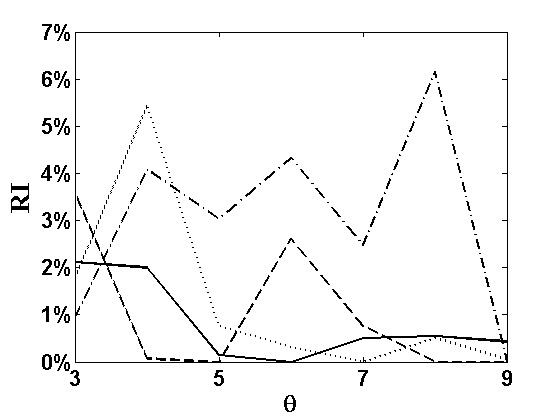}
	\label{fig:RI_n10_h05_l-5}}
~ 
\subfloat[$n=10$, $\lambda=-0.5$, $h^*=1$]{
	\includegraphics[width=0.23\textwidth]{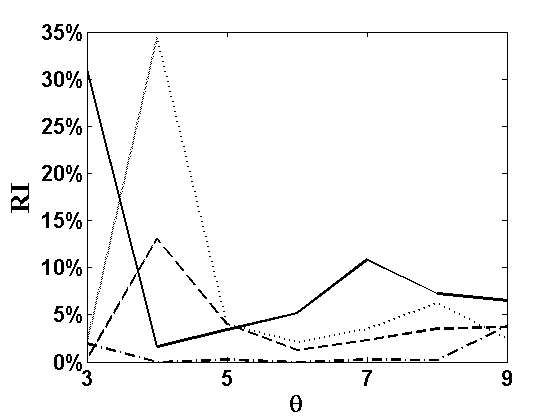}
	\label{fig:RI_n10_h10_l-5}}
~ 
\subfloat[$n=20$, $\lambda=-0.5$, $h^*=0.5$]{
	\includegraphics[width=0.23\textwidth]{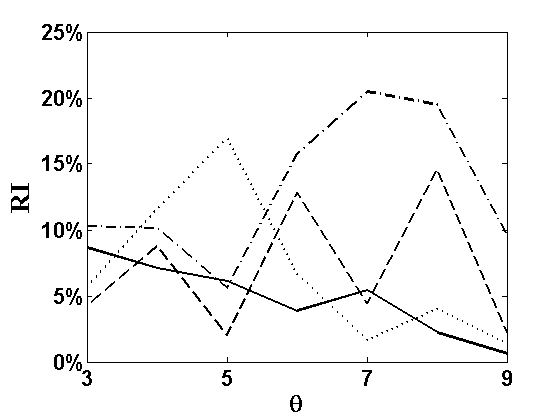}
	\label{fig:RI_n20_h05_l-5}}
~ 
\subfloat[$n=20$, $\lambda=-0.5$, $h^*=1$]{
	\includegraphics[width=0.23\textwidth]{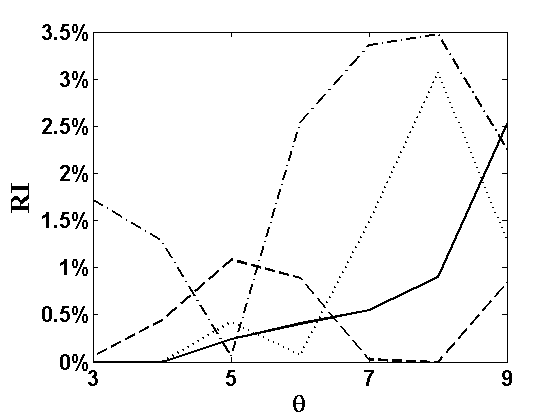}
	\label{fig:RI_n20_h10_l-5}}
	\\ 
\subfloat[$n=10$, $\lambda=-1$, $h^*=0.5$]{
	\includegraphics[width=0.23\textwidth]{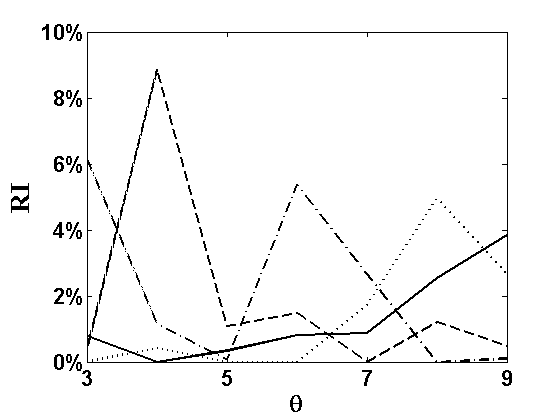}
	\label{fig:RI_n10_h05_l-10}}
~ 
\subfloat[$n=10$, $\lambda=-1$, $h^*=1$]{
	\includegraphics[width=0.23\textwidth]{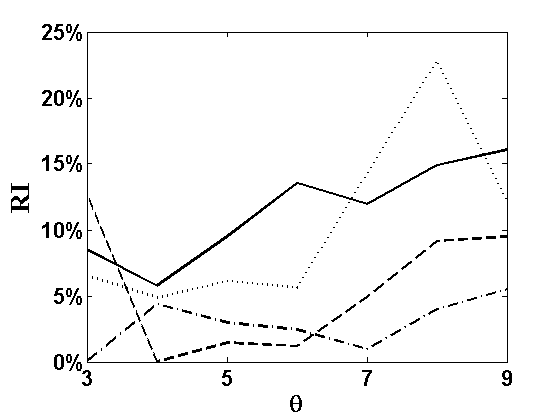}
	\label{fig:RI_n10_h10_l-10}}
~ 
\subfloat[$n=20$, $\lambda=-1$, $h^*=0.5$]{
	\includegraphics[width=0.23\textwidth]{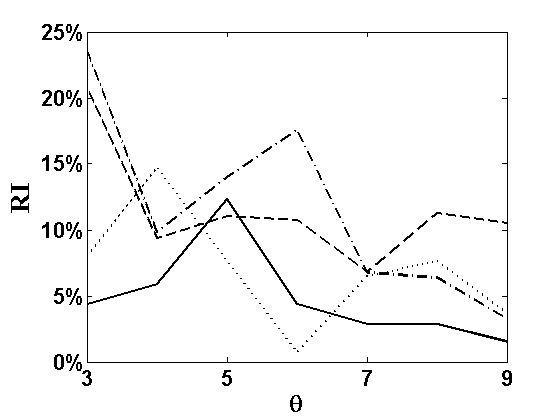}
	\label{fig:RI_n20_h05_l-10}}
~ 
\subfloat[$n=20$, $\lambda=-1$, $h^*=1$]{
	\includegraphics[width=0.23\textwidth]{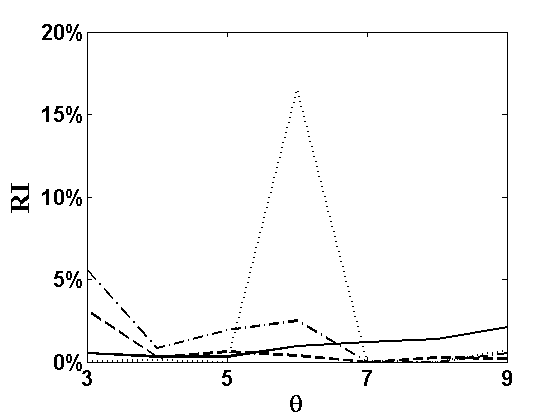}
	\label{fig:RI_n20_h10_l-10}}
	\\
\subfloat[$n=10$, $\lambda=-1.5$, $h^*=0.5$]{
	\includegraphics[width=0.23\textwidth]{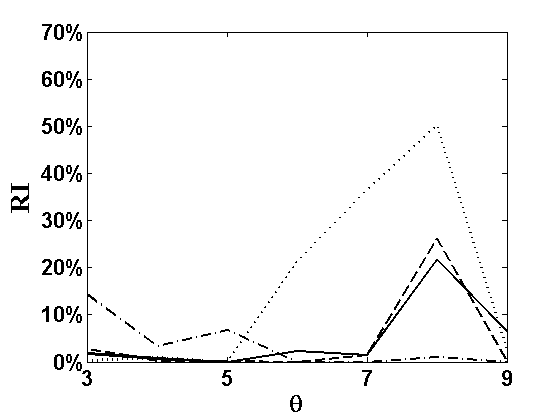}
	\label{fig:RI_n10_h05_l-15}}
~ 
\subfloat[$n=10$, $\lambda=-1.5$, $h^*=1$]{
	\includegraphics[width=0.23\textwidth]{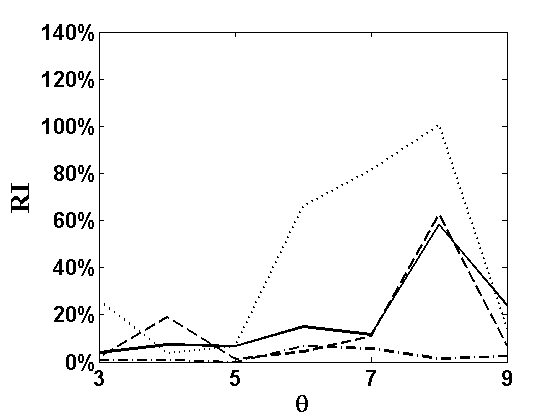}
	\label{fig:RI_n10_h10_l-15}}
~ 
\subfloat[$n=20$, $\lambda=-1.5$, $h^*=0.5$]{
	\includegraphics[width=0.23\textwidth]{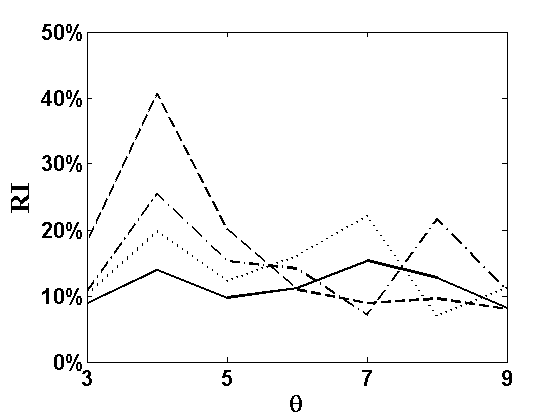}
	\label{fig:RI_n20_h05_l-15}}
~ 
\subfloat[$n=20$, $\lambda=-1.5$, $h^*=1$]{
	\includegraphics[width=0.23\textwidth]{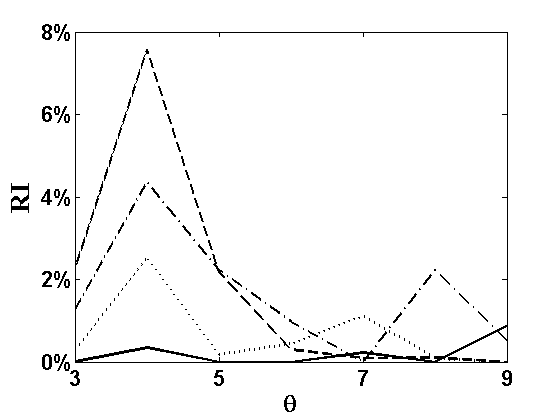}
	\label{fig:RI_n20_h10_l-15}}
	\\ 
\subfloat[$n=10$, $\lambda=-2$, $h^*=0.5$]{
	\includegraphics[width=0.23\textwidth]{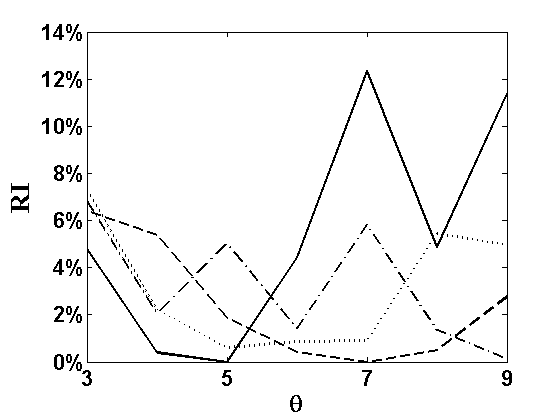}
	\label{fig:RI_n10_h05_l-20}}
~ 
\subfloat[$n=10$, $\lambda=-2$, $h^*=1$]{
	\includegraphics[width=0.23\textwidth]{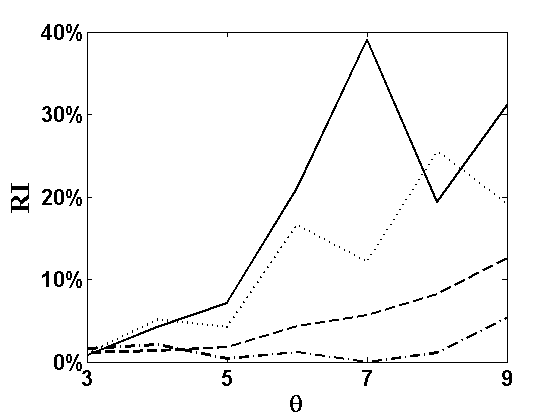}
	\label{fig:RI_n10_h10_l-20}}
~ 
\subfloat[$n=20$, $\lambda=-2$, $h^*=0.5$]{
	\includegraphics[width=0.23\textwidth]{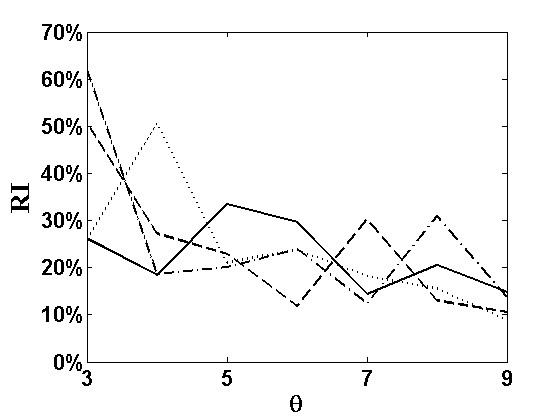}
	\label{fig:RI_n20_h05_l-20}}
~ 
\subfloat[$n=20$, $\lambda=-2$, $h^*=1$]{
	\includegraphics[width=0.23\textwidth]{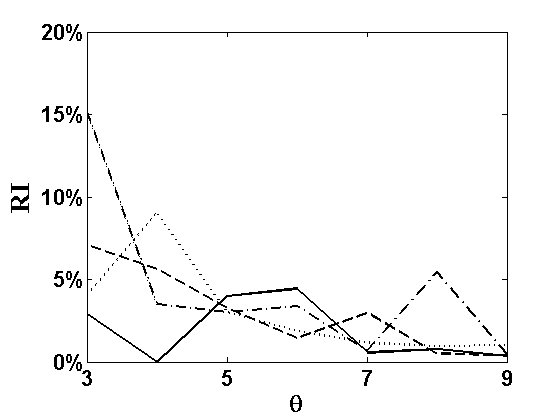}
	\label{fig:RI_n20_h10_l-20}}
	\caption{Relative increase (RI) in MSE due to use of the simple choice $h^*$ over the optimum choice 
	for different $\lambda$ and $\alpha$ at sample sizes $n=10, 20$ 
	(solid line: $\alpha=0$, dotted line: $\alpha=0.1$, dashed line: $\alpha=0.25$ and dash-dotted line: $\alpha=0.5$). }
	\label{FIG:5RI_MPSDE_n1020}
\end{figure}

As we have already observed, the MSE of the MPSDE does not vary appreciably in the interval $[0.5, 1.5]$, 
and $h =0.5$ or $h = 1$ is often a ``close to optimal" choice. 
Based on this observation, we next study the relative increase in MSE (and hence loss in efficiency) 
of the estimators for each of the cases considered previously for the simpler choices $h^*=0.5$ or $1$
instead of the optimum $h$, say $h_{opt}$; we define this measure as
$$
\mbox{ RI} = \frac{\mbox{MSE}(h^*) - \mbox{MSE}(h_{opt})}{\mbox{MSE}(h_{opt})}.
$$
Figure \ref{FIG:5RI_MPSDE_n1020} plots these measures of relative increase in MSE for the sample sizes $n=10$ and $20$;
the results for $n=15$ are similar with MSE values in between these two and hence omitted to save space.
It can be clearly observed from the figure that, when $\lambda=0$, 
the natural weight 1 ($=1/A$) in fact does not lead to an RI of more than $5\%$; 
note that this is somewhat expected since the inlier issue is not very serious for $\lambda=0$ 
and is more important for $\lambda<0$ (see Figure  \ref{FIG:5MSE_MSDE}).
Further, for such estimators with $\lambda<0$, in most cases,
we generally do not have more than $10\%$ relative increase in MSE 
while using $h^* = 0.5$ for $n=10$ and  using $h^* = 1$ for $n=20$.
These choices, therefore, give the practitioners a guidance 
on a simple primary application which can be refined at a later stage 
using more detailed exploration of the role of $h$. 
The cases where we have a larger percentage of relative increase correspond 
to very small values of the mean square error, both penalized and ordinary.

\begin{table}[!th]
	\centering 
	\caption{Parameter estimates for the Drosophila data with different methods, 
		along with the true empty cell weight = $\frac{1}{A}$}
	\begin{tabular}{lr|c|ccc} \hline
		$\lambda$	& $\alpha$	&  $1/A$	& MSDE$^{(a)}$	&	MPSDE($h=0.5$)	&	MPSDE($h=1$)	 \\\hline 
		0	&	0	&	1	&	3.0588	&	3.3873	&	3.0588	\\	
		0	&	0.1	&	1	&	0.3917	&	0.3998	&	0.3917	\\	
		0	&	0.25	&	1	&	0.3858	&	0.3905	&	0.3858	\\	
		0	&	0.5	&	1	&	0.3747	&	0.3763	&	0.3747	\\	\hline
		$-$0.5	&	0	&	2	&	0.3637	&	0.3945	&	0.3829	\\	
		$-$0.5	&	0.1	&	1.82	&	0.3704	&	0.3902	&	0.3820	\\	
		$-$0.5	&	0.25	&	1.60	&	0.3732	&	0.3829	&	0.3783	\\	
		$-$0.5	&	0.5	&	1.33	&	0.3696	&	0.3723	&	0.3707	\\	\hline
		$-$1	&	0	&	$\infty$	&	--	&	0.3831	&	0.3714	\\	
		$-$1	&	0.1	&	10	&	0.2955	&	0.3803	&	0.3722	\\	
		$-$1	&	0.25	&	4	&	0.3491	&	0.3754	&	0.3709	\\	
		$-$1	&	0.5	&	2	&	0.3638	&	0.3684	&	0.3668	\\	\hline
		$-$1.5	&	0	&	$-$2	&	--	&	0.3716	&	0.3601	\\	
		$-$1.5	&	0.1	&	$-$2.86	&	--	&	0.3706	&	0.3627	\\	
		$-$1.5	&	0.25	&	$-$8	&	--	&	0.3681	&	0.3638	\\	
		$-$1.5	&	0.5	&	4	&	0.3549	&	0.3647	&	0.3632	\\	\hline
		$-$2	&	0	&	$-$1	&	--	&	0.3608	&	0.3498	\\	
		$-$2	&	0.1	&	$-$1.25	&	--	&	0.3615	&	0.3540	\\	
		$-$2	&	0.25	&	$-$2	&	--	&	0.3614	&	0.3573	\\	
		$-$2	&	0.5	&	$\infty$	&	--	&	0.3613	&	0.3598	\\	\hline
		
	\end{tabular}\\
\small{$^{(a)}$ It is `--' when the $S$-divergence is not defined due to the presence of empty cells}
	\label{TAB:real_data}
\end{table}

\section{A Real Data Example}\label{SEC:examples}

we now present a real life application of the proposed minimum penalized $S$-divergence estimators.
We consider a segment of the Drosophila data \citep{Woodruff/etc:1984} based on a chemical mutagenicity experiment. 
The dataset contains the number of daughters carrying a recessive lethal mutation on their X chromosome 
among (roughly) 100 sampled daughter flies from each male Drosophila fly 
when exposed to a certain level of a chemical and mated with unexposed female flies 
in a particular (on day 177) experimental run.
The observed frequencies of the male flies are $\boldsymbol{r}_n = (23, 7, 3, 1)$ having 
$x = (0, 1, 2, 91)$ recessive lethal daughters; all other values of $x$ has frequency zero.
Clearly there is one large outlier in the data (at 91) and plenty of empty cells.
The dataset can be modeled nicely by a Poisson model except for the outlying point,
as described in \cite{Simpson:1987}, \cite{Basu/etc:2011} and \cite{Ghosh:2013}. 
The last paper presented the minimum $S$-divergence estimators for the Poisson mean parameter 
with these data both with and without the outlier, 
where it was observed that the $S$-divergences with negative $\lambda$ yield robust estimators. 
But some of these MSDEs (including the Hellinger distance) are substantially affected by the presence of empty cells
in the data and hence differ significantly from the outlier deleted MLE (which is 0.3939).
In fact, many robust members of the $S$-divergence family with $\lambda<-1$ 
are not finitely defined for this dataset due to the presence of empty cells and 
hence the corresponding estimators cannot be obtained.

However, we can obtain the proposed minimum penalized $S$-divergence estimators
of the Poisson parameter $\theta$ for this dataset at any value of the tuning parameters $(\alpha,\lambda)$.
These MPSDEs are reported in Table \ref{TAB:real_data} for the suggested simple choices of $h=0.5,1$,
along with the corresponding MSDE whenever defined.
In terms of the matching of the observed and expected data (excluding the outliers), 
the estimates in the (0.38, 0.39) window appear to perform the best. 
However, the estimators with natural penalty weight are sometimes shifted by a large
amount from this region due to the empty cell effect. A case in point is the $(\lambda=-1, \alpha=0.1)$ combination,
where the natural estimator is drastically affected by the empty cells, but the penalties put them in the desired zone.
An even stronger effect of this phenomenon (not presented in Table \ref{TAB:real_data}) 
is for the $(\lambda=-0.9, \alpha=0)$ combination.

\section{Conclusions and Discussions}\label{SEC:discussion}

Many minimum divergence estimators, including those within the class of disparities 
and the class of S-divergences, have excellent robustness properties, 
but are often handicapped in small samples due to their poor inlier controlling properties 
which may lead to substantially degraded model performance. 
An extreme form of inliers are the empty cells, and suitable empty cell corrections are useful
in improving this small sample model performance. 
In the tradition of research on this topic, 
we believe that we have made some significant additions to the literature. 
Our achievements and recommendations are listed below.

\begin{itemize}
	\item The divergences which have natural empty cell penalty factor equal to $\infty$ 
	cannot be ordinarily defined in an infinite sample space. 
	But with our penalized scheme there is no problem with their construction; 
	our approach also fills in the theoretical convergence and distributional properties of such estimators, 
	so far unavailable in the literature.  
	
	\item Our approach generalizes the inlier controlling strategy beyond the class of disparities.  
	
	\item We have provided actual (simulation based) figures of the  optimal penalty factor 
	for different values of the Poisson mean parameter $\theta$, different $(\alpha, \lambda)$ combinations, 
	and a few small to moderate sample sizes. While this is fairly detailed, a completely automatic, 
	case specific recommendation for the penalty factor to choose in a given situation may require additional research. 
	However, as an overall recommendation, we observe that in general $h \in [0.5, 1]$ works well in practically all situations. 
\end{itemize}

\begin{figure}[!b]
	\centering
	\subfloat[$n=10$, $\theta=3$]{
		\includegraphics[width=0.33\textwidth]{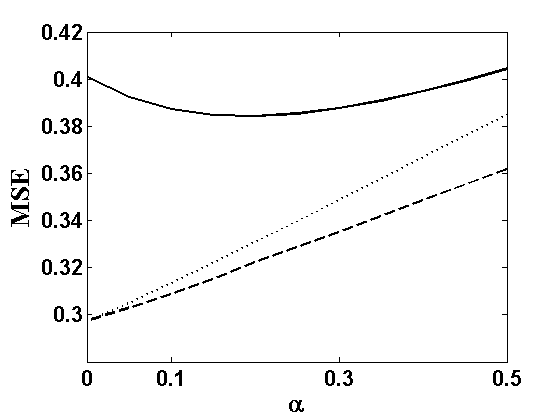}
		\label{fig:opt_MSE_n10_theta3_l-05}}
	~ 
	\subfloat[$n=10$, $\theta=5$]{
		\includegraphics[width=0.33\textwidth]{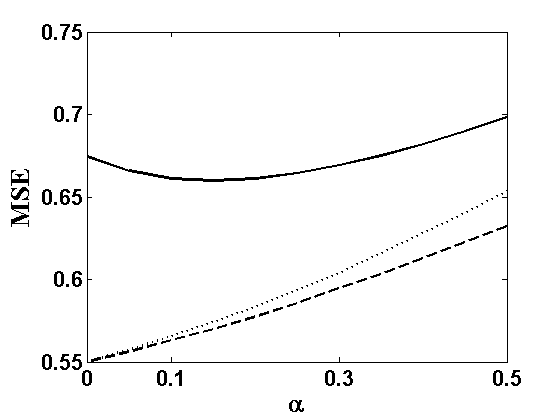}
		\label{fig:RI_n10_h10_l0}}
	~ 
	\subfloat[$n=10$, , $\theta=9$]{
		\includegraphics[width=0.33\textwidth]{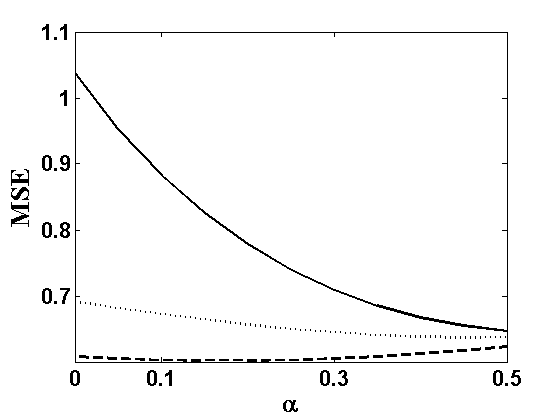}
		\label{fig:RI_n10_h10_l0}}
	\\		
	\subfloat[$n=20$, $\theta=3$]{
		\includegraphics[width=0.33\textwidth]{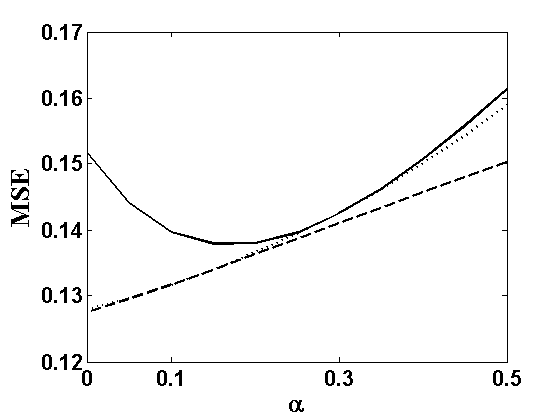}
		\label{fig:opt_MSE_n10_theta3_l-05}}
	~ 
	\subfloat[$n=20$, $\theta=5$]{
		\includegraphics[width=0.33\textwidth]{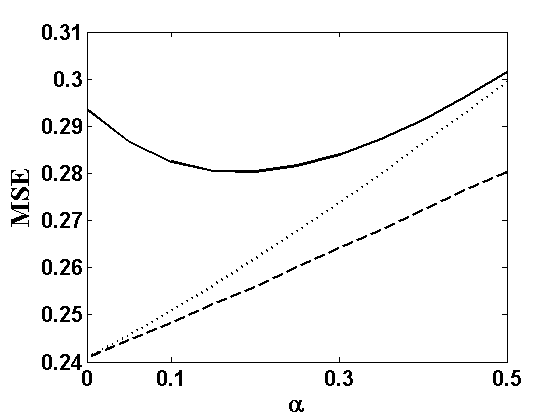}
		\label{fig:RI_n10_h10_l0}}
	~ 
	\subfloat[$n=20$, , $\theta=9$]{
		\includegraphics[width=0.33\textwidth]{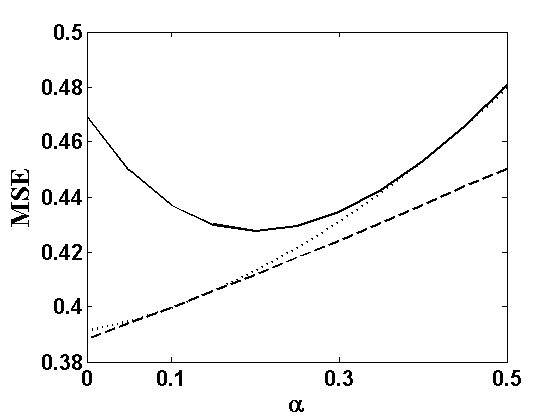}
		\label{fig:RI_n10_h10_l0}}
	\caption{Empirical MSEs of the MPSDEs with $\lambda=-0.5$ over $\alpha$,  
		with the natural empty cell weight $h=1/A$ (solid line), 
		with the optimum $h$ obtained under the formulation of (\ref{EQ:PSD}) (dotted line)
		and with the optimum $h$ and $\beta$ obtained under the formulation of (\ref{EQ:PSD2}) (dashed line)
		for different sample size $n$ and model parameters $\theta$. }
	\label{FIG:new_beta}
\end{figure}

We will end the paper with the mention of another possible extension of the definition of the PSD. 
In the definition (\ref{EQ:PSD}), it is possible to add one more tuning parameter $\beta$ 
to replace the parameter $\alpha$  in the term corresponding to the empty cells,
so that a modified penalized $S$-divergence measure can be defined as 
\begin{eqnarray}
PSD_{(\alpha, \lambda)}^{h,\beta}(\mathbf{r}_n,\mathbf{f}_\theta) &=& ~ \sum_{x:r_n(x)\ne 0} ~ 
\left[ \frac{1}{A} f_\theta^{1+\alpha}(x)  -   \frac{1+\alpha}{A B} ~ f_\theta^{B}(x) r_n^{A}(x)  
+ \frac{1}{B} ~ r_n^{1+\alpha}(x) \right] \nonumber\\
&& ~~  +  h \sum_{x:r_n(x)=0} f_\theta^{1+\beta}(x). 
\label{EQ:PSD2}
\end{eqnarray}
One can again define the MPSDE based on this new definition of the Penalized S-Divergence as given in Equation (\ref{EQ:PSD2}) 
of the PSD and it will follow, along the same lines of the proof given in Section \ref{SEC:5MPSDE_asymp_discrete},
that these modified MPSDEs also have the same asymptotic properties as the previous version 
given in Theorem \ref{THM:5discrete_asymp_MPSDE}. As the intuitive motivation suggests,
one may possibly achieve a better inlier control by varying both $h$ and $\beta$ simultaneously 
(for any given divergence with fixed $\alpha$ and $\lambda$) generating estimators with even smaller MSEs.
For a brief illustration, in Figure \ref{FIG:new_beta},
we have plotted the resulting optimum MSE (minimum MSE over both $h$ and $\beta$)
along with the optimum MSE obtained for definition (\ref{EQ:PSD}) (minimum MSE over $h$ only)
and the MSE obtained by the natural empty cell weight $1/A$ for $\lambda=-0.5$
and different $\alpha$, $\theta$ and $n$; 
the pattern is similar for other $\lambda<0$ and hence not reported here.

We have not developed this analysis to the extent where we can make a definite recommendation 
about the value of the $\beta$ parameter to use in a given situation. 
However Figure \ref{FIG:new_beta} gives ample evidence of the fact that 
there is the possibility of further improving the small sample performance of the MSDEs, 
and it may be worthwhile to further pursue the role of the $\beta$ parameter. 

As a final word we point out that as all modifications involving the $h$ and $\beta$ parameters relate to the inliers, 
the improvement that is obtained in either case is achieved without compromising the outlier stability properties of the divergence. 
This is not just a technical observation, we have noticed this repeatedly in our simulations. 
However, we have not actually put up such tables in the paper that illustrate the robust behavior of the MSPDEs, 
as our interest here is on improving small sample model efficiency, 
rather than exploring the robustness of the estimators.

%
%

\end{document}